\documentclass[sigconf, nonacm]{acmart}

\newcommand\vldbdoi{XX.XX/XXX.XX}
\newcommand\vldbpages{XXX-XXX}
\newcommand\vldbvolume{14}
\newcommand\vldbissue{1}
\newcommand\vldbyear{2020}
\newcommand\vldbauthors{\authors}
\newcommand\vldbtitle{\shorttitle} 
\newcommand\vldbavailabilityurl{https://github.com/tarlaun/hifive}
\newcommand\vldbpagestyle{plain} 

\usepackage{etoolbox}
\AtBeginEnvironment{table}{\footnotesize}

\usepackage{amsmath}
\usepackage{subcaption}

\usepackage{pgfplots}
\usepackage{pgfplotstable}
\usepgfplotslibrary{groupplots}
\usepgfplotslibrary{statistics} 
\pgfplotsset{compat=1.18}       
\usepackage{tikz}

\graphicspath{{},{figs/}}

\usepackage{algorithm}
\usepackage[noend]{algpseudocode}
\usepackage{booktabs}
\usepackage{multirow}
\usepackage{siunitx}
\usepackage{enumitem}

\usepackage{hyperref}
\usepackage{xspace}
\usetikzlibrary{matrix}
\usepackage{xstring}
\usetikzlibrary{patterns}

\usepackage{balance}

\newcommand{\hifive}{HiFIVE\xspace}

\newcommand{\mycomment}[1]{}

\title{\hifive: High-Fidelity Vector-Tile Reduction for Interactive Map Exploration}

\author{Tarlan Bahadori}
\affiliation{%
  \institution{University of California, Riverside}
  \city{Riverside}
  \country{USA}
}
\email{tbaha001@ucr.edu}

\author{Ahmed Eldawy}
\affiliation{%
  \institution{University of California, Riverside}
  \city{Riverside}
  \country{USA}
}
\email{eldawy@ucr.edu}
\begin{abstract}
Interactive visualization is a common tool for exploring large open-data repositories, where users quickly explore datasets across diverse domains. When it comes to large-scale spatial data, many existing tools rely on server-side rendering to produce small images that can be viewed at the client-side. However, most users prefer client-side rendering that allows quick styling of the data for better visualization experience.
This paper presents \hifive, a data-management framework for scalable, high-fidelity client-side geospatial visualization. We formalize the visualization-aware tile reduction problem, which captures the trade-off between tile-size and visualization distortion, and prove its NP-hardness. \hifive introduces a two-stage solution combining triage and sparsification to selectively prune records, attributes, and values based on information-theoretic and spatial criteria. Experiments demonstrate substantial tile-size reductions while preserving visual fidelity and interactive performance at terabyte scale.
\end{abstract}

\begin{document}

\maketitle
\mycomment{
\pagestyle{\vldbpagestyle}
\begingroup\small\noindent\raggedright\textbf{PVLDB Reference Format:}\\
\vldbauthors. \vldbtitle. PVLDB, \vldbvolume(\vldbissue): \vldbpages, \vldbyear.\\
\href{https://doi.org/\vldbdoi}{doi:\vldbdoi}
\endgroup
\begingroup
\renewcommand\thefootnote{}\footnote{\noindent
This work is licensed under the Creative Commons BY-NC-ND 4.0 International License. Visit \url{https://creativecommons.org/licenses/by-nc-nd/4.0/} to view a copy of this license. For any use beyond those covered by this license, obtain permission by emailing \href{mailto:info@vldb.org}{info@vldb.org}. Copyright is held by the owner/author(s). Publication rights licensed to the VLDB Endowment. \\
\raggedright Proceedings of the VLDB Endowment, Vol. \vldbvolume, No. \vldbissue\ %
ISSN 2150-8097. \\
\href{https://doi.org/\vldbdoi}{doi:\vldbdoi} \\
}\addtocounter{footnote}{-1}\endgroup
}
\ifdefempty{\vldbavailabilityurl}{}{
\vspace{.3cm}
\begingroup\small\noindent\raggedright\textbf{Artifact Availability:}\\
The source code, data, and/or other artifacts have been made available at \url{\vldbavailabilityurl}.
\endgroup
}

\section{Introduction}


Interactive visualization is one of the fastest, and sometimes the only, means to determine what a dataset actually contains and assess its fitness to answer certain scientific questions. With hundreds of thousands of publicly available datasets~\cite{data_gov,world_bank_open_data,undata}, users must be able to scan, compare, and eliminate candidates quickly in order to select the right data for the task at hand. Moreover, multiple reports estimate that roughly 60\% of publicly available data carries a geospatial component, highlighting the central role of maps as a primary interface for exploration across a wide range of domains, including urban planning~\cite{DZM+18,GFC+18}, public health~\cite{CIZ+22}, environmental monitoring~\cite{KZM+19}, and transportation~\cite{YLT+15}.


Scaling geospatial exploration to today's data volumes remains difficult. Although web maps use multi-resolution tile pyramids for interactive pan-and-zoom~\cite{aidstar,hadoopviz,ucrstar}, existing systems still face a core trade-off between scalability and visualization fidelity. Existing approaches can be largely classified as \emph{server-side rendering} or \emph{client-side rendering}. In \emph{server-side rendering}, the server generates and styles \emph{raster tiles}, which scale well because computation happens on the back end and tile size is bounded by pixel resolution~\cite{aidstar,geosparkviz,JLZ+16}. However, the visualization is effectively fixed, and the client cannot restyle or recover details that were not encoded in the pixels. In \emph{client-side rendering}, the server delivers \emph{vector tiles}~\cite{10.14778/3750601.3750690,tremmel2025mlt,GFC+18,SGM+13}, i.e., geometries and attributes, and the client performs styling, enabling high-fidelity, fully customizable maps. Yet vector tiles are not naturally size-bounded. Payload can grow arbitrarily with feature density, geometry complexity, and attribute cardinality, often exceeding practical network and browser rendering budgets. Some work has been done to use dictionary-style encoding in Mapbox Vector Tiles (MVT)~\cite{mvt} or column compression in Maplibre Tiles (MLT)~\cite{tremmel2025mlt}, but they do not address the unbounded tile size problem.

\begin{figure*}[t]
    \centering

    \begin{subfigure}[t]{0.33\textwidth}
        \centering
        \includegraphics[height=3.6cm,keepaspectratio]{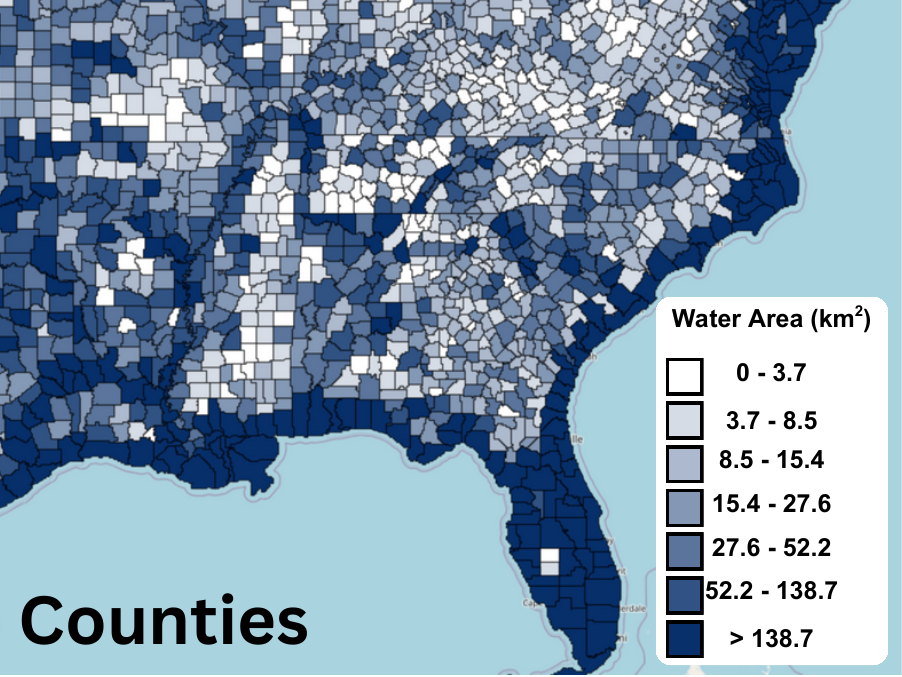}
        \caption{Counties color-coded by water area}
        \label{fig:counties-awater}
    \end{subfigure}\hfill
    \begin{subfigure}[t]{0.33\textwidth}
        \centering
        \includegraphics[height=3.6cm,keepaspectratio]{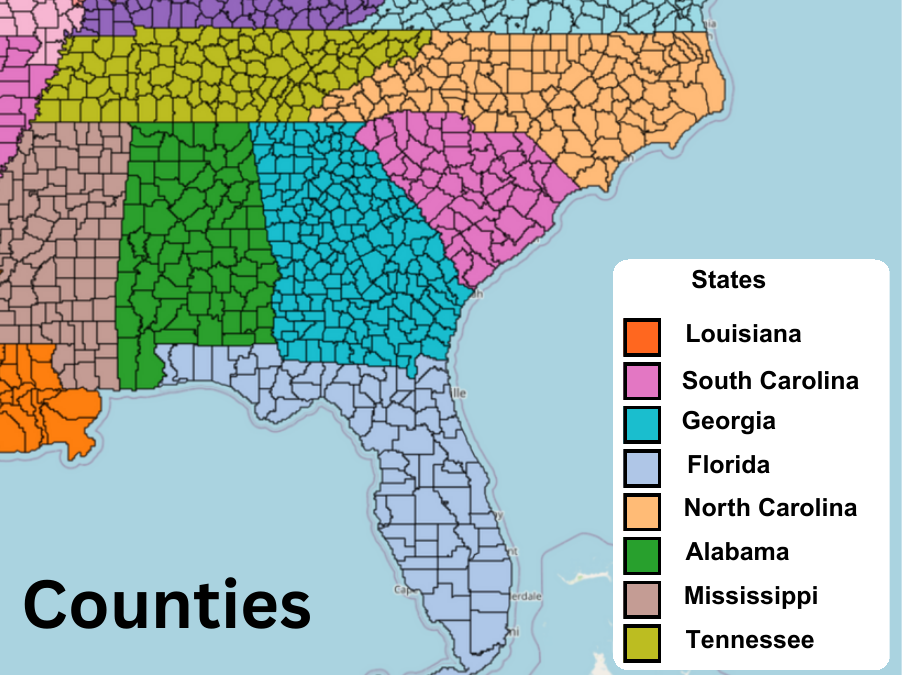}
        \caption{Counties color-coded by state (\hifive)}
        \label{fig:counties-hi5}
    \end{subfigure}\hfill
    \begin{subfigure}[t]{0.33\textwidth}
        \centering
        \includegraphics[height=3.6cm,keepaspectratio]{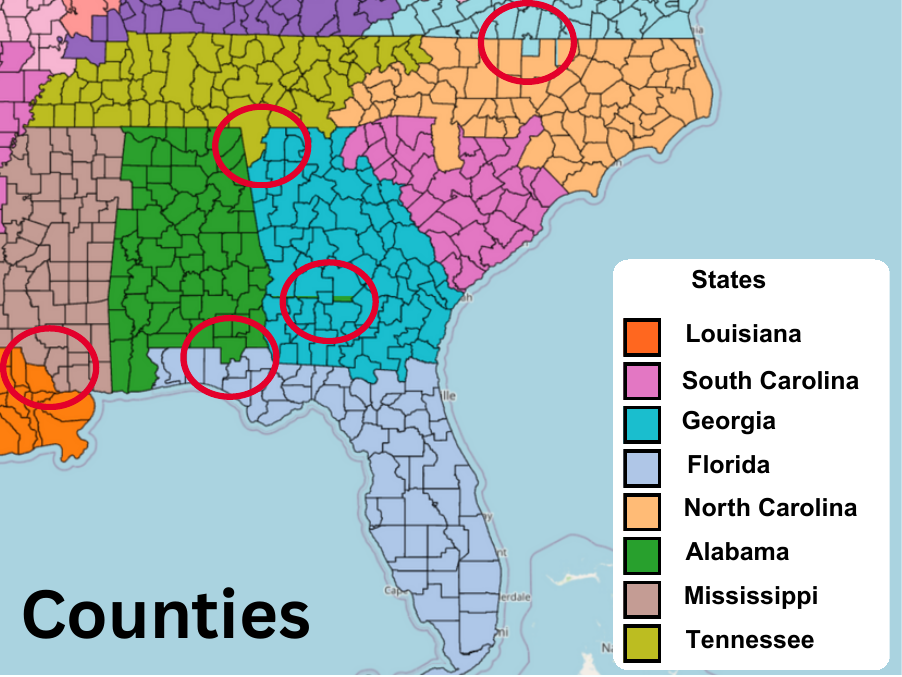}
        \caption{Counties colored by state (Tippecanoe~\cite{tippecanoe})}
        \label{fig:counties-tip}
    \end{subfigure}

    \vspace{0.8ex}

    \begin{subfigure}[t]{0.33\textwidth}
        \centering
        \includegraphics[height=3.6cm,keepaspectratio]{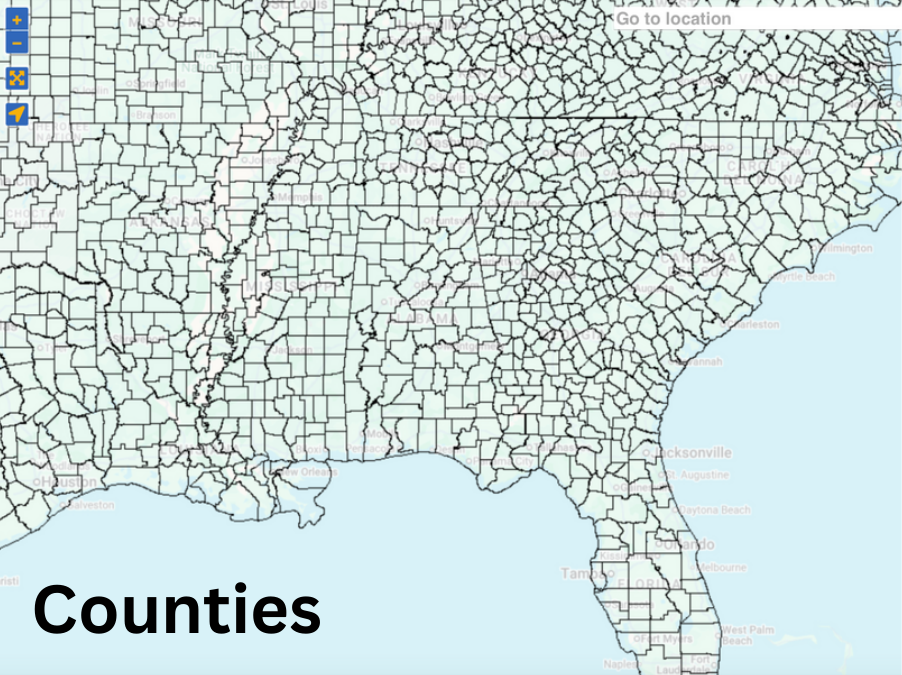}
        \caption{Unstyled counties (UCR-STAR~\cite{ucrstar})}
        \label{fig:counties-aid}
    \end{subfigure}\hfill
    \begin{subfigure}[t]{0.33\textwidth}
        \centering
        \includegraphics[height=3.6cm,keepaspectratio]{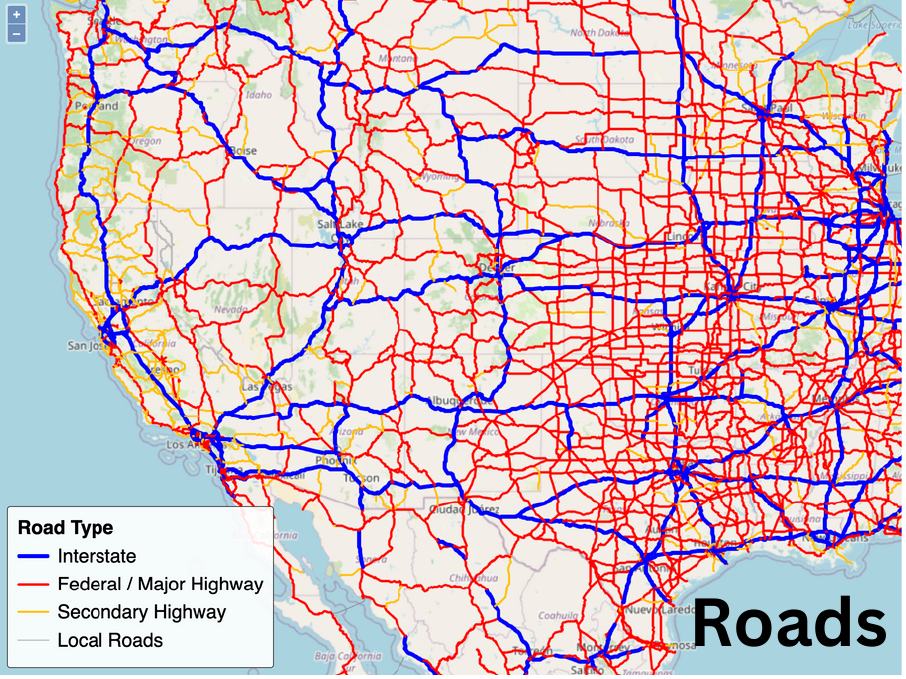}
        \caption{Roads color-coded by road type}
        \label{fig:roads}
    \end{subfigure}\hfill
    \begin{subfigure}[t]{0.33\textwidth}
        \centering
        \includegraphics[height=3.6cm,keepaspectratio]{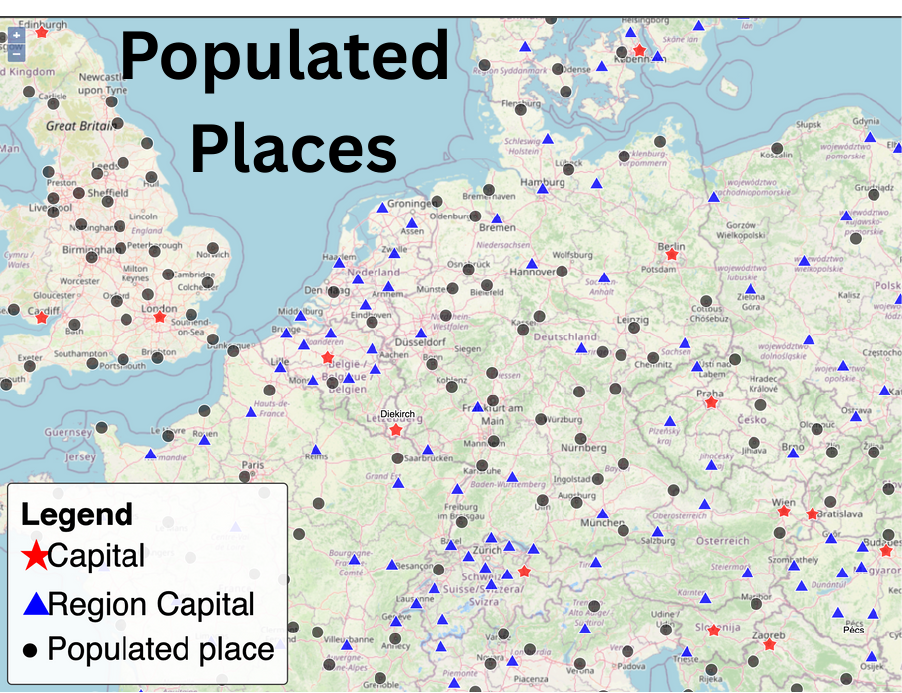}
        \caption{Populated places styled by their type}
        \label{fig:pop-places}
    \end{subfigure}

    \caption{Different client-side styles applied to the same set of vector tiles.}
    \label{fig:main}
\end{figure*}

\autoref{fig:main} illustrates client-side rendering for three datasets. Figures~\ref{fig:main}(a-d) display the US Counties dataset published by the Census Bureau~\cite{UCRSTAR/TIGER2018/COUNTY}. \autoref{fig:counties-awater} shows how a hydrologist might color-code counties by water area coverage. In Figures \ref{fig:counties-hi5} and \ref{fig:counties-tip} an administrator might color-code the same dataset by state. However, it can be seen in \autoref{fig:counties-tip} how some existing tools, e.g., Tippecanoe~\cite{tippecanoe}, produces a low-fidelity visualization where state boundaries are distorted. The reason is that Tippecanoe tries to reduce the data size by merging nearby polygons which causes inaccurate visualization. On the other hand, \autoref{fig:counties-aid} shows AID*~\cite{aidstar} visualization with server-side rendering which does not allow any of these styles. Figures~\ref{fig:roads} and~\ref{fig:pop-places} show two other examples with roads and populated cities which illustrate how users might style data in various ways requiring efficient client-side rendering.


This paper presents \hifive, a \underline{hi}gh-\underline{fi}delity \underline{v}isualization system for map \underline{e}xploration that resolves two fundamental bottlenecks in interactive vector-tile visualization, \emph{scalability} and \emph{visualization fidelity}. As data gets large, the resulting vector tiles can be massive, and, unlike raster tiles, vector tiles have no inherent upper bound, a challenge that is particularly imminent when dealing with tens of thousands of datasets. Aggressively and naively shrinking tiles is only valuable if the resulting maps remain visually and semantically faithful. \hifive reframes tile generation as a data-management optimization problem. The key idea is that not all records and not all attributes are equally important to produce a high-fidelity visualization. For example, a spatially large polygon is more prominent than a small one in the visualization. Additionally, some attributes, e.g., categorical attributes, tend to take less space and are more useful in styling compared to, say, a generated UUID, which consumes space and provides little value for visualization. \hifive exploits these asymmetries by deciding which records and attributes to retain at each tile such that tiles remain within a certain size while supporting rich client-side styling. Since the decision is made independently per tile, tiles become more detailed as the user zooms in and the tile size decreases.


This paper introduces and formalizes the \emph{visualization-aware tile reduction} problem. It builds on concepts from information theory to measure the information loss in reduced tiles and combines them with spatial data metrics to capture visualization loss. To build on that, we propose a two-stage solution, \emph{sparsification} and \emph{triage}. The sparsification step reduces the tile size by setting some values to null while minimizing information and visual loss. The triage step applies a set of heuristics to make quick and significant reductions to the tile, making it ready for the find-grained sparsification step. Together, these techniques produce compact, standard-compliant vector tiles that support interactive, user-defined styling entirely on the client, enabling exploration over terabyte-scale datasets without incurring per-style recomputation on the server. We validate our approach through extensive experiments, demonstrating substantial reductions in tile size while preserving visual fidelity and maintaining interactive performance at scale.

Below, we summarize the contributions of this paper:
\begin{itemize}[leftmargin=12pt,topsep=0pt]
    \item We propose the tile distortion metric that quantifies visual information loss in map tiles.
    \item We formalize the visualization-aware tile reduction problem and prove its NP-hardness.
    \item We introduce \hifive with a two-stage solution to the problem based on mixed integer linear programming.
    \item We carry out an extensive experimental evaluation to study the scalability of \hifive.
\end{itemize}

The remainder of the paper is organized as follows.
\autoref{sec:related_work} describes the related work on spatial data visualization.
\autoref{sec:overview} summarizes our proposed framework and outlines its key components.
\autoref{sec:modeling} defines the problem of visualization-aware tile reduction and proves its NP-hardness.
\autoref{sec:milp-size-model} presents the tile sparsification algorithm using mixed integer linear programming.
\autoref{sec:triage} describes the triage step.
\autoref{sec:experiments} reports experimental results demonstrating the effectiveness of our approach.
Finally, \autoref{sec:conclusion} concludes the paper and describes future work.

\section{Related Work} \label{sec:related_work}

Geospatial visualization systems face dual challenges of managing exponentially growing datasets while maintaining interactive performance. In recent years, two dominant approaches have emerged to tackle these challenges:

First, \textbf{cluster-based rendering} systems like GeoSparkViz~\cite{geosparkviz}, HadoopViz~\cite{hadoopviz}, and AID*~\cite{aidstar} employ \textbf{server-side rendering} through parallel computing frameworks such as Spark or Hadoop MapReduce. While effective for generating high-resolution static maps, their architecture limits client-side interactivity as the tiles are partially or fully pregenerated at the server side. In addition, separate approaches have to be designed for each type of visualization such as kernel density visualization~\cite{CIZ+22} and heat maps~\cite{LKS13}. These approaches can scale to big data but the server gets quickly overloaded as users request different styling options.



The second approach is \textbf{client-side rendering} which focuses on delivering the raw data to the client which can style and render it based on user needs without incurring any additional overhead on the server. Some tools first aggregate the data at the server side and deliver a summarized version that can be visualized efficiently but it still requires the server to be aware of the visualization type and style~\cite{JLZ+16}. Another approach is to sample the data for visualization~\cite{vas} but this approach is limited to point data and does not consider custom styles such as color-coded points. Other approaches build vector tiles, e.g., MVT~\cite{mvt} or MLT~\cite{tremmel2025mlt}, which contain all records with all their attributes. Unlike raster tiles, vector tiles are not size-bounded and can grow arbitrarily large, posing an additional overhead on transferring and rendering these tiles. Planetiler~\cite{planetiler} proposes a parallel algorithm for generating map tiles but it uses application-specific rules that are specific to OpenStreetMap data. Similarly, \cite{SGM+13} introduces the \emph{thinning} problem to cap the number of records per tile in Google Maps while maintaining geographic coherence across tiles but it does not consider the importance of individual attributes. Tippecanoe~\cite{tippecanoe} provides a general purpose MVT tile generator that employs a greedy, density-based feature dropping technique to ensure map legibility across zoom levels. However, its heuristic-driven simplification disregards the semantic utility of non-spatial attributes, leading to visual fidelity loss when styling is applied to data.

This paper follows the client-side rendering approach but makes a key shift in treating tile reduction as a first-class optimization problem, not an ad-hoc side problem. Given a user-specified size budget, \hifive selects which records, attributes, and values to materialize in each tile to minimize information loss while preserving visualization semantics. This yields compact tiles that can be styled on the client side with a tunable level-of-detail parameter.

\begin{figure}[t]
    \centering
    \includegraphics[width=\linewidth]{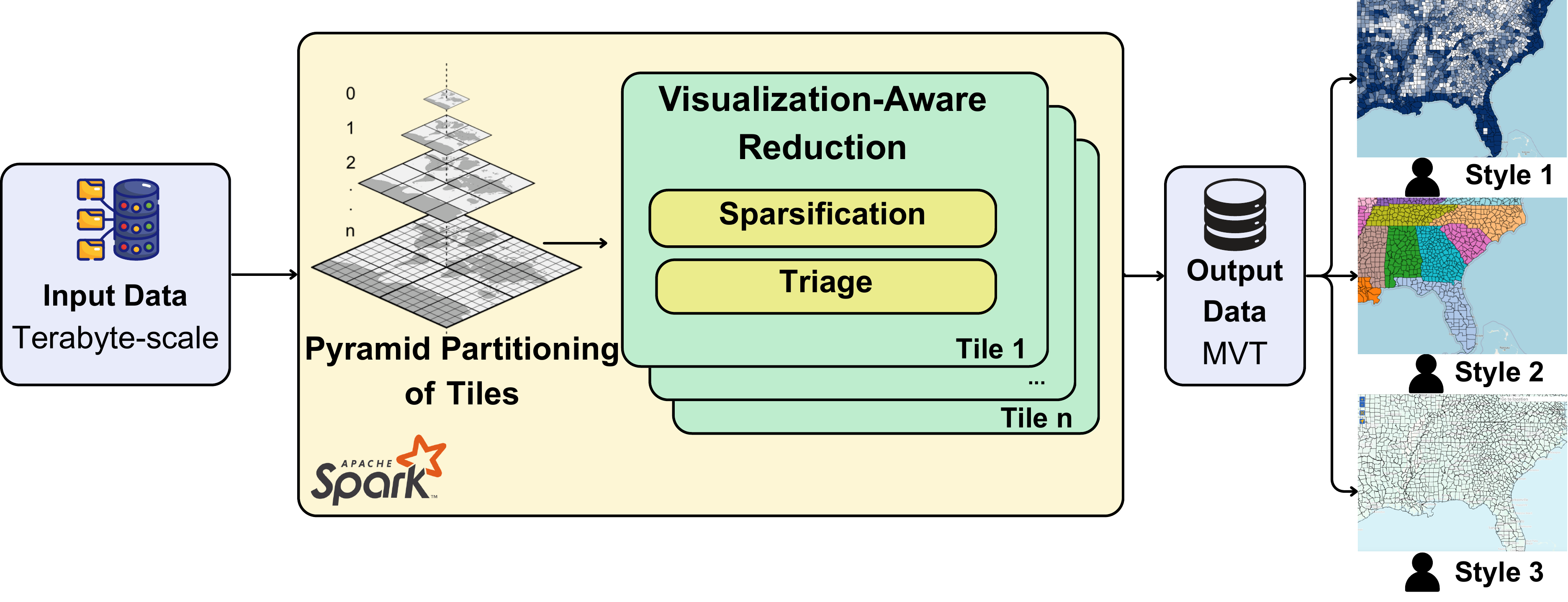} 
    \caption{This is an overview of the \hifive pipeline.}
    \label{fig:overview} 
\end{figure}

\section{Overview}
\label{sec:overview}

This work enables scalable, high-fidelity geospatial visualization via a two-stage data processing pipeline that produces compact, stylable vector tiles as shown in \autoref{fig:overview}. Our goal is to support terabyte-scale datasets while enforcing per-tile size budget that allows swift network transfer and smooth client-side rendering in modern web mapping applications.

The pipeline begins with a \emph{data preparation stage} that (i) reprojects input geometries to the web Mercator projection used in web maps and (ii) partitions the data into a multi-zoom tile pyramid structure compatible with standard web visualization libraries~\cite{OpenLayers,tremmel2025mlt,googlemaps}. This stage follows well-established practice and we refer the readers to~\cite{aidstar,hadoopviz} for full details.

Unlike raster tiles~\cite{aidstar,hadoopviz,geosparkviz}, vector tiles store individual geometries and their attributes. This results in tile sizes that are highly variable and can exceed browser and network budgets in dense regions. Therefore, this paper focuses on \emph{visualization-aware tile reduction}, that is, given a tile and a user-specified size budget $B$, it aims at reducing the tile size so it can be serialized to no more than $B$ bytes while minimizing loss in visualization fidelity. Further details on the problem definition are given in \autoref{sec:modeling}.

To solve the visualization-aware tile reduction problem, we run two algorithms individually on each tile, \emph{triage} and {\em sparsification}, where the triage algorithm makes big cuts to bring the tile size down while the sparsification algorithm is the primary one that brings it down exactly to the desired size.

This sparsification algorithm models the problem as a linear program (LP) and uses an LP solver to simultaneously reduce tile size while maintaining visualization fidelity. The key idea is to set certain attribute values to null with the goal of reducing the data size and minimizing the effect on visualization. This is done through a set of constraints and an objective function that balance (i) visual prominence measured in pixel footprint, (ii) storage size, and (iii) information loss measured via KL-divergence. Further details on this algorithm are given in \autoref{sec:milp-size-model}.

To ensure scalability on extremely dense tiles, we also introduce a lightweight \emph{visualization-aware triage} algorithm which uses various empirical techniques to reduce the tile size before running the more costly sparsification step. The triage step involves removing selected rows and columns and quantizing numeric columns. The order of operations is determined by \emph{attribute saliency}, measured via the drop in entropy caused by reduction. This ensures that triage targets attributes least likely to be used for styling the output visualizations later on. Further details on this step are given in \autoref{sec:triage}.

Both algorithm run in parallel using a distributed Spark job which allows horizontal scale up with data size. They output standard MVT tile~\cite{mvt} that can be rendered using any existing client-side visualization library such as OpenLayers~\cite{OpenLayers} or MapLibre~\cite{tremmel2025mlt}. Since our reduction algorithms do not alter the tile schema or format, existing styling tools can be used without modification.

\section{Visualization-Aware Tile Reduction Problem Definition}
\label{sec:modeling}


\hifive adopts the Mapbox Vector Tile (MVT) file format for storing all tiles. This allows clients, e.g., web browsers, to style the data based on user input without any overhead on the server~\cite{10.14778/3750601.3750690}. To enable client-side rendering, each tile contains all records in the geographic region that the tile covers. This data is stored in a table with $N$ rows and $d$ columns where each row represents a map feature, i.e., record, and each column represents an attribute, e.g., spatial geometry, country name, and population. Unfortunately, this structure cannot scale to very large datasets, as tiles will become excessively large for network transfer and client-side rendering.


The problem we address is how to reduce the size of these tiles while preserving the semantic and visual fidelity of the final rendered map. More specifically, given an input tile, how to produce an output tile that represents the data in that region but fits within the desired tile size.

This section starts with some preliminary definitions and then follows with the formal problem definition. We use the simple example tile in \autoref{tab:toy8} to help illustrate the definitions.


\begin{table}[t]
\footnotesize
\centering
\renewcommand{\arraystretch}{1.1}
\caption{Toy example with $N=4$ lakes and $d=3$ attributes. The salinity attribute indicates if a lake is (f)resh or (s)alt.}
\label{tab:toy8}

\begin{subtable}[t]{0.48\linewidth}
\centering
\caption{Input tile $T_{\text{in}}$}
\label{tab:toy8-input}
\begin{tabular}{c l c}
\toprule
\textbf{Geometry} & \textbf{name} & \textbf{salinity} \\
\midrule
G1 & Azul   & f \\
G2 & Birch  & s \\
G3 & Cobalt & s \\
G4 & Dune   & f \\
\bottomrule
\end{tabular}
\end{subtable}
\hfill
\begin{subtable}[t]{0.48\linewidth}
\centering
\caption{Output tile $T_{\text{out}}$}
\label{tab:toy8-output}
\begin{tabular}{c l c}
\toprule
\footnotesize
\textbf{Geometry} & \textbf{name} & \textbf{salinity} \\
\midrule
G1 & Azul & f \\
G2 & Birch & s \\
G3 & $\bot$ & s \\
G4 & $\bot$ & $\bot$ \\
\bottomrule
\end{tabular}
\end{subtable}

\end{table}



\begin{definition}[Schema ($\boldsymbol{S}$)]
\label{def:schema}
A \emph{schema} ($\boldsymbol{S}$) is an ordered set of attribute types
$\boldsymbol{S} = (\tau_1,\dots,\tau_d)$
where each attribute (column) $j\in\{1,\dots,d\}$ is associated with its own \emph{domain} (carrier set) $\text{Dom}_j$, which is a \emph{finite} set of admissible values for that specific column:
\[
\text{Dom}_j \subseteq \text{Dom}(\tau_j), \qquad |\text{Dom}_j| < \infty.
\]
Thus, even if $\tau_j=\tau_k$ (e.g., two \textsf{Int} attributes), the corresponding domains $\text{Dom}_j$ and $\text{Dom}_k$ may differ.
The domain always contains the null value $\bot$. Without loss of generality, we assume that the first attribute is always the geometry attribute, i.e., $\tau_1=$\texttt{Geometry}.
\end{definition}

In \autoref{tab:toy8}, the schema is (\texttt{Geometry}, \texttt{String}, \texttt{String}). The domain of the \texttt{name} attribute $\mathrm{Dom}_2=\{$Azul, Birch, Cobalt, Dune, $\bot\}$ which is limited to the values in this tile. Even though the name and salinity attributes are both of type \texttt{String}, their domains differ based on their values.

\begin{definition}[Feature]
\label{def:feature}
Given a schema $\boldsymbol{S}$, a \emph{feature} ($f$) is a typed $d$-tuple
$f \in \prod_{j=1}^{d}\text{Dom}_j$,
where the $j$-th attribute value is referenced positionally as $f[j] \in \text{Dom}_j$.
\end{definition}

\begin{definition}[Tile]
\label{def:tile}
Given a schema $\boldsymbol{S}$, a \emph{tile} ($T$) is a finite set of features
\[
T = \{f_i\}_{i=1}^{N}, \qquad f_i \in \prod_{j=1}^{d}\text{Dom}j
\]
We write $T \in \text{Rel}(\boldsymbol{\tau})$ to denote that $T$ is a relation over schema $\boldsymbol{S}$. The cardinality of $T$ is the number of features in the tile, $N=|T|$.
\end{definition}

In \autoref{tab:toy8-input}, the input tile has four features, i.e.,  $|T_{\text{in}}|=4$.
The output tile $T_{\text{out}}$ has the same schema and cardinality,
i.e., $|T_{\text{out}}|=|T_{\text{in}}|=4$, but differs in its attribute values.

\begin{definition}[Column Size]
\label{def:colsize}
Given a tile $T$, the $j$-th \emph{column} is the sequence of its $j$-th attribute values:
\[
C_j(T) = \langle t_1[j],\dots,t_N[j]\rangle \in \text{Dom}(\tau_j)^N
\]
Let $\text{Enc}_j$ be a type-dependent column encoding function that jointly encodes all values in column $j$:
\[
\text{Enc}_j:\text{Dom}(\tau_j)^N \rightarrow \{0,1\}^*
\]
The \emph{size} of column $j$ is defined as the number of bytes in its encoded representation:
\[
\text{Size}_j(T) \;=\; \bigl|\text{Enc}_j(C_j(T))\bigr|_{\mathrm{bytes}}
\]
This models columnar formats, e.g., MLT~\cite{tremmel2025mlt}, and dictionary-style formats, e.g., MVT~\cite{mvt}, where values in a column are encoded together rather than independently per row.
\end{definition}

In the example, nullifying attribute values reduces the number of distinct non-null entries and increases the frequency of null ($\bot$), leading to a smaller encoded representation for the affected columns.

\begin{definition}[Tile Size]
\label{def:tilesize}
The \emph{size} of a tile $T$ is the sum of the sizes of all its columns:
\[
\text{Size}(T) \;=\; \sum_{j=1}^{d}\text{Size}_j(T)
\]
\end{definition}

\noindent\textbf{Pixel-Weighted Semantic Divergence.}
The following part defines how two tiles $T_1$ and $T_2$ diverge from each other based on their values and perceived rendering. Our goal is to minimize the distortion between the two tiles which captures:
(1) \textbf{Visual Salience:} Attribute values belonging to features that occupy more screen space, i.e., have larger pixel footprints, are more visually prominent; and
(2) \textbf{Semantic Impact:} We should minimize the information loss, represented by value distributions, in the attributes that are likely to be used for styling. We observe that attributes with lower entropy, e.g., categorical attributes, are more likely to be used in styling since they have an observable pattern.

\begin{figure}[t]
  \centering
  \def\svgwidth{0.95\linewidth}
  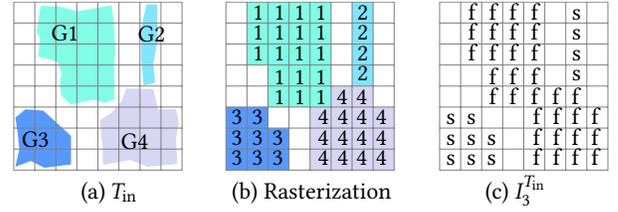
  \caption{Rasterization example of the tile $T_\text{in}$. (a) The input geometries. (b) The rasterization of the geometries. (c) The rasterization of salinity attribute ($j=3$). The empty cells denote $\bot$.}
  \Description{A figure with three subfigures indicating the input geometries on the left, their rasterization in the middle, and the rasterization of the salinity attribute on the right. }
  \label{fig:rasterize}
\end{figure}

\begin{definition}[Rendering Function ($\mathcal{R}$)]
\label{def:render}
Given a tile resolution of $R$ pixels and a geometry $g$, a \emph{rendering (or rasterization) function} returns the set of pixels that the geometry $g$ covers.
\begin{equation}
\mathcal{R}: g \rightarrow \{1,\dots,R\}
\label{eqn:render}
\end{equation}
\end{definition}
\autoref{fig:rasterize} shows an example of rasterizing the four geometries with $R=64$ pixels. We indicate the result of the rasterize function with the row number, i.e., 1 through 4.

\begin{definition}[Attribute Image]
\label{def:attimage}
Given a tile $T$, a rendering function $\mathcal{R}$, and a non-geometry attribute $j\in\{2,\dots,d\}$ we define an \emph{attribute image} $I_j^T$:

\[
I_j^{T} : \{1,\dots,R\} \rightarrow \text{Dom}_j
\]
by assigning each pixel the attribute value of the feature whose geometry covers it:
\begin{equation}
I_j^{T}(\mathbf{x}) \;=\;
\begin{cases}
f_i[j] & \text{if } \mathbf{x} \in \mathcal{R}(f_i[1]) \text{ for some feature } i,\\
\bot   & \text{if no geometry covers } \mathbf{x}.
\end{cases}
\end{equation}
Thus, $I_j^{T}$ can be viewed as a synthetic image obtained by \emph{painting} each geometry with its attribute value.
\end{definition}

\autoref{fig:rasterize}(c) shows the image of the salinity attribute for $T_\text{in}$ where each pixel is marked with either `\texttt{f}', `\texttt{s}', or left empty to indicate null.

\begin{table}[t]
\centering
\footnotesize
\renewcommand{\arraystretch}{1.15}
\caption{Pixel-weighted counts $c_j^{T}(v)$ and corresponding Laplace-smoothed
distributions $\tilde p_j^{T}(v)$ for the toy example ($R=64$, $\varepsilon=1$).
\label{tab:toy8-dists}
Uncovered pixels are assigned to $\bot$.
Bold values indicate changes between $T_{\text{in}}$ and $T_{\text{out}}$.}
\begin{tabular}{l l r r c c}
\toprule
\textbf{Attribute $j$} & \textbf{Value $v$}
& $c_j^{T_{\text{in}}}(v)$ & $c_j^{T_{\text{out}}}(v)$
& $\tilde p_j^{T_{\text{in}}}(v)$ & $\tilde p_j^{T_{\text{out}}}(v)$ \\
\midrule
\multirow{5}{*}{name}
 & Azul   & 18 & 18 & 0.275362 & 0.275362 \\
 & Birch  & 14 & 14 & 0.217391 & 0.217391 \\
 & Cobalt & 8  & \textbf{0}  & 0.130435 & \textbf{0.014493} \\
 & Dune   & 4  & \textbf{0}  & 0.072464 & \textbf{0.014493} \\
 & $\bot$ & 20 & \textbf{32} & 0.304348 & \textbf{0.478261} \\
\midrule
\multirow{3}{*}{salinity}
 & fresh  & 32 & 32 & 0.492537 & 0.492537 \\
 & salt   & 12 & \textbf{8}  & 0.194030 & \textbf{0.134328} \\
 & $\bot$ & 20 & \textbf{24} & 0.313433 & \textbf{0.373134} \\
\bottomrule
\end{tabular}

\vspace{0.4em}
For \texttt{name}, $|\mathrm{Dom}_j|=5$ and the normalization constant is
$R + \varepsilon|\mathrm{Dom}_j| = 64 + 5 = 69$.
For \texttt{salinity}, $|\mathrm{Dom}_j|=3$ and the normalization constant is
$64 + 3 = 67$.
\end{table}

\begin{definition}[Pixel-weighted Empirical Distribution]
\label{def:weighted-dist}
The pixel-weighted count of a value $v\in\text{Dom}_j$ is the number of pixels in $I_j^T$ whose value equals $x$:
\begin{equation}
c_j^{T}(v) \;=\; \left|\left\{\, \mathbf{x} \in \{1,\dots,R\} \;:\; I_j^{T}(\mathbf{x}) = v \,\right\}\right|
\label{eq:distributions}    
\end{equation}
\autoref{tab:toy8-dists} shows the pixel counts for the two tiles used in the example. Notice that even though the values `\texttt{f}' and `\texttt{s}' each appear twice in the table, their corresponding pixel-weighted counts are $c_3^{T\text{in}}(s)=12$ and $c_3^{T\text{in}}(f)=32$.

Based on these counts, we compute the empirical distribution of values while also applying \emph{Laplace smoothing} with parameter $\varepsilon>0$ added to each count $c_j^{T}$ to account for zero counts:
\begin{equation}
\tilde p_j^{T}(v) \;=\; \frac{c_j^{T}(v) + \varepsilon}{R + \varepsilon|\text{Dom}_j|},
\qquad v\in\text{Dom}_j
\label{eqn:laplace}
\end{equation}
\end{definition}

\autoref{tab:toy8-dists} shows the adjusted empirical weights. Notice how zero counts become small adjusted weights due to smoothing.

\begin{definition}[Visualization-aware Attribute Divergence]
\label{def:VAD}
Given two tiles $T_1$ and $T_2$ with an identical schema and number of rows, we compute the visualization-aware divergence of an attribute $j\in\{2,\dots,d\}$ as the \emph{Jensen-Shannon divergence} between their pixel-weighted empirical distributions. First, we define the mixture distribution:
\begin{equation}
m_j=\tfrac{1}{2}\left(\tilde p_j^{T_1}+\tilde p_j^{T_2}\right)
\label{eqn:mixture}
\end{equation}

Based on that, the visualization-aware divergence is defined as:
\begin{equation}
\mathrm{VAD}_j(T_1,T_2)=\tfrac{1}{2} \mathrm{KLD}\left(\tilde p_j^{T_1} \,\|\, m_j\right)+\tfrac{1}{2} \mathrm{KLD}\left(\tilde p_j^{T_2} \,\|\, m_j\right)
\label{eqn:VAD}
\end{equation}
where:
\begin{equation}
\mathrm{KLD}(p\|q)=\sum_{x\in\text{Dom}_j} p(x)\log_2\frac{p(x)}{q(x)} \label{eqn:KLD}  
\end{equation}

Note that using $\log_2$ ensures that $\textrm{VAD}\in[0,1]$ for all columns $j$.
\end{definition}

\autoref{tab:toy8-metrics} shows the divergence between the name and salinity columns between the two sample tiles.

\paragraph{Interpretation.}
Because the distributions are induced by the painted images $I_j^{T}$, any change in attribute values for geometries with large pixel footprints moves more probability mass and therefore increases $\mathrm{VAD}_j$ more strongly. In particular, nullifying a value for a visually large geometry transfers mass from its original value bin to the $\bot$ bin, yielding a corresponding larger divergence.

Next, we compute the divergence between two tiles by aggregating the divergence of all columns. However, we observe that users tend to apply styling using attributes with low entropy, e.g., categorical attributes with a few unique values. In the example, the `salinity' attribute is has a lower entropy than the `name' attribute (unique value for each record) and unlikely to be used for styling the visualization. Thus, we use inverse-entropy-weighted aggregation to compute the tile divergence as shown below.

\begin{definition}[Tile Distortion]
\label{def:TLD}
Given two tiles $T_1$ and $T_2$ with identical schema, we define the \emph{tile distortion} as an inverse-entropy-weighted aggregation of per-attribute divergences. We define the entropy of attribute $j$ in $T_1$ as:
\begin{equation}
H_j(T_1) \;=\; -\sum_{x\in \text{Dom}_j} \tilde p_j^{T_1}(x)\,\log_2 \tilde p_j^{T_1}(x).
\label{eqn:entropy}
\end{equation}
To emphasize attributes with lower entropy, we assign each attribute a weight that is inversely related to its entropy:
\[
w_j(T_1) \;=\; \frac{(H_j(T_1)+\delta)^{-\gamma}}{\sum\limits_{k=2}^{d} (H_k(T_1)+\delta)^{-\gamma}},
\]
where $\delta>0$ prevents division by zero and $\gamma\ge 0$ controls the strength of the emphasis on low-entropy attributes (with $\gamma=0$ yielding uniform weights).
The tile-level distortion is then:
\begin{equation}
\begin{aligned}
    \mathrm{TLD}(T_1,T_2) \;=\; \sum_{j=2}^{d} w_j(T_1)\cdot \mathrm{VAD}_j(T_1,T_2).
\end{aligned}
\label{eq:tld}
\end{equation}

\end{definition}

\begin{table}[t]
\centering
\footnotesize
\renewcommand{\arraystretch}{1.15}
\caption{Computed entropy $H_j(T_{\text{in}})$, per-attribute divergence $\mathrm{VAD}_j(T_{\text{in}},T_{\text{out}})$ (JSD, log$_2$),
inverse-entropy weights $w_j(T_{\text{in}})$ with $\delta=10^{-9}$, $\gamma=1$, and final $\mathrm{TLD}$.}
\label{tab:toy8-metrics}
\begin{tabular}{l c c c}
\toprule
\textbf{Attribute $j$} &
$H_j(T_{\text{in}})$ (bits) &
$\mathrm{VAD}_j(T_{\text{in}},T_{\text{out}})$ &
$w_j(T_{\text{in}})$ \\
\midrule
name     & 2.170965 & 0.067751 & 0.406485 \\
salinity & 1.486845 & 0.005812 & 0.593515 \\
\bottomrule
\end{tabular}

\vspace{0.5em}
\begin{tabular}{c}
\toprule
$\displaystyle
\mathrm{TLD}(T_{\text{in}},T_{\text{out}})
= \sum_{j\in\{\texttt{name},\,\texttt{salinity}\}}
w_j(T_{\text{in}})\cdot \mathrm{VAD}_j(T_{\text{in}},T_{\text{out}})
= 0.030989
$ \\
\bottomrule
\end{tabular}
\vspace{-15pt}
\end{table}

\autoref{tab:toy8-metrics} completes the example by showing the entropy of each column, the adjusted weights, and the final tile distortion.

\noindent\textbf{Visualization-aware Tile Reduction Problem Definition.}
Given an input tile $T_{\text{in}}$ with schema $\boldsymbol{S}$ and a byte budget $B$, our goal is to construct an output tile $T_{\text{out}}$ (with the same schema) that fits within the budget while minimizing visual/semantic distortion relative to $T_{\text{in}}$. Formally, we solve:
\begin{equation}
\label{eq:tile-reduction}
\begin{aligned}
\min_{T_{\text{out}} \in \text{Rel}(\boldsymbol{S})} \quad & \mathrm{TLD}(T_{\text{in}}, T_{\text{out}}) \\
\text{s.t.} \quad & \mathrm{Size}(T_{\text{out}}) \le B, \\
& |T_{\text{out}}| = |T_{\text{in}}| \;\;
\end{aligned}
\end{equation}

\begin{theorem}
\label{thm:nphard}
The visualization-aware tile reduction problem defined in \autoref{eq:tile-reduction} is NP-hard.
\end{theorem}

\begin{proof}[Proof sketch]
We prove this theorem by showing a polynomial-time reduction from the 0-1 Knapsack problem. Consider an instance of the 0-1 Knapsack with $n$ items of sizes $s_i$, values $v_i$, and capacity $B$. We construct an instance of the tile reduction problem with an input tile $T_{in}$ with $n$ features and two attributes. The first attribute is geometry and the second attribute has a unique value for each feature, e.g., $f_i^{\text{in}}[2]=i$. Further, we adjust the rendering function so that each feature is rasterized to exactly $v_i$ pixels with no overlap and no empty pixels. This can be easily achieved if each geometry is a line of length $v_i$ and the total number of pixels is $\sum v_i$. Then, we restrict the values in the output tile $T_{out}$ to be either the same corresponding value as $T_{in}$ or null ($\bot$), i.e., $f_i^{\text{out}}[2]\in\{f_i^{\text{in}}[2],\bot\}$. Finally, we define the tile encoder function such that the size of each attribute is equal to the corresponding size from the Knapsack problem if it is not null, i.e.,
$\text{Size}(T_{\text{out}})=\sum s_i \;:\;f_i^{\text{out}}[2]\neq\bot$.

This setting ensures that the size of non-null values stays within the Knapsack capacity $B$ which is straight-forward. At the same time, it minimizes the total value of nullified values which, maximizes the value of the non-null values, as shown below.

First, any non-null, i.e., retained, value $i$ in $T_{\text{out}}$ will contribute zero to KLD since $p_i^{T_\text{in}}(x)=p_i^{T_{\text{out}}}(x)=m(x)$ as shown in Equations~\ref{eqn:mixture} and \ref{eqn:KLD}. On the other hand, if the value $i$ is nullified, all its mass gets assigned to the null value ($\bot$), i.e., $p_i^{T_\text{out}}=0$ and $m_i=\tfrac{1}{2}p_i^{T_\text{in}}$. Therefore, its term in \autoref{eqn:KLD} will be:
\[p_i^{T_\text{in}}\log_2\left(\frac{p_i^{T_\text{in}}}{\tfrac{1}{2}p_i^{T_\text{in}}}\right)=p_i^{T_\text{in}}\log_2(2)=p_i^{T_\text{in}}\]
Since the number of pixels of each feature \(i\) equals its Knapsack value \(v_i\), and these values are unique, we have \(p_i^{T_\text{in}} = v_i\). Moreover, because the input tile has no empty pixels (i.e., \(p^{T_\text{in}}_\bot = 0\)), the null value does not contribute to the KLD. Therefore, tile reduction minimizes the total size of nulled values, which is equivalent to maximizing the total size of retained (non-null) values. Hence, it solves the Knapsack problem, completing the proof.
\end{proof}

Notice that the proof above requires the use of non-smoothed distributions, i.e. $\varepsilon=0$ in \autoref{eqn:laplace} but this does not change the nature of the problem.

\section{MILP-based Sparsification}
\label{sec:milp-size-model}

Given a vector tile containing $N$ records (rows) and $d$ attributes (columns), our goal is to reduce the serialized tile size below a user-specified threshold $B$ while preserving visual salience and attribute utility. This problem applies to every tile in the pyramid that is above the desired threshold. Due to the problem complexity, we simplify and reformulate it as a mixed-integer linear program (MILP) that jointly decides (i) which records to retain and (ii) which attribute values (cells) to retain for the kept records.

\paragraph{Decision variables.}
Without loss of generality, we assume that the first column contains the geometry attribute.
We introduce the following binary variables:
\begin{align}
y_i &\in \{0,1\} && \text{if geometry of record $i\in[1,N]$ is kept} \\
u_j &\in \{0,1\} && \text{if non-geometric attribute $j\in[2,d]
$ is kept} \\
x_{i,j} &\in \{0,1\} && \text{if attribute cell $(i\in[1,N],j\in[2,d])$ is kept}
\end{align}
Variables $x_{i,j}$ are instantiated only for non-null cells and $u_j$ are instantiated if the column is not empty.

\paragraph{Structural constraints.}
A cell can be kept only if both its geometry and its column are kept:
\begin{align}
x_{i,j} &\le y_i && \forall i\in[1,N],j\in[2,d], \\
x_{i,j} &\le u_j && \forall i\in[1,N],j\in[2,d].
\end{align}

\begin{table}[t]
\centering
\renewcommand{\arraystretch}{1.15}
\footnotesize
\begin{minipage}[t]{0.47\columnwidth}
\centering
\caption{Example MILP variables for input $T_{in}$ (table \ref{tab:toy8-input}).}
\label{tab:milp-x-symbolic}
\begin{tabular}{c c c}
\toprule
  & \textbf{name ($u_2$)} & \textbf{salinity ($u_3$)} \\
\midrule
$y_1$ & $x_{1,2}$ & $x_{1,3}$ \\
$y_2$ & $x_{2,2}$ & $x_{2,3}$ \\
$y_3$ & $x_{3,2}$ & $x_{3,3}$ \\
$y_4$ & $x_{4,2}$ & $x_{4,3}$ \\
\bottomrule
\end{tabular}
\end{minipage}
\hfill
\begin{minipage}[t]{0.47\columnwidth}
\centering
\footnotesize
\caption{Example MILP solution; results in $T_{\text{out}}$ (table \ref{tab:toy8-output}).}
\label{tab:milp-x-inst}
\begin{tabular}{c c c}
\toprule
  & \textbf{1} & \textbf{1} \\
\midrule
1 & 1 & 1 \\
1 & 1 & 1 \\
1 & 0 & 1 \\
1 & 0 & 0 \\
\bottomrule
\end{tabular}
\end{minipage}
\end{table}

\paragraph{Linear size model.}
To approximate the tile size for our optimization constraint, we decompose the total storage into three primary components: geometry data, fixed column overhead, and an amortized cost per cell. Let $\widehat{b}^{geom}_i$ denote the estimated bytes contributed by the geometry of record $i$, derived from its vertex count. In vector-tile encodings, attribute values are not stored independently; instead, each layer maintains a shared dictionary of distinct values. Consequently, the storage cost of an attribute column $j$ consists of a dictionary cost $\widehat{b}^{dict}_j$, representing the bytes required for unique non-null values, and a per-cell pointer cost $b^{ptr}$ used to reference the dictionary.

To maintain a linear model, we distribute the dictionary cost across the $n_j$ non-null cells in the column. Thus, when a cell $(i,j)$ is retained, it contributes an estimated cost of $b^{ptr} + \widehat{b}^{dict}_j / n_j$. While removing a single cell only reduces the dictionary size if a unique value is entirely eliminated, this approximation remains highly effective. Columns with high cardinality naturally result in higher amortized costs, making them primary candidates for reduction under a fixed size budget. Conversely, columns with few distinct values contribute minimally to the total size, ensuring that the approximation has negligible impact on the overall constraint. Future work could involve refining this linear approximation to better capture the discrete jumps in dictionary size inherent in the NP-hard visualization-aware tile reduction problem.
The resulting linear size constraint is:
\begin{equation}
\sum_{i=1}^{N} \widehat{b}^{geom}_i\,y_i
+
\sum_{(i,j)}
\left(
b^{ptr} + \frac{\widehat{b}^{dict}_j}{n_j}
\right)
x_{i,j}
\;\le\;
B,
\label{eq:milp-size-constraint}
\end{equation}
where $B$ is the target tile size. 

\paragraph{Objective function.}
The MILP maximizes a weighted utility that balances visual salience of records and fidelity of retained attributes:
\begin{equation}
\max
\;\;
\alpha\sum_{i=1}^{N} U_i^{rec}\,y_i
+
(1-\alpha)\sum_{(i,j)} U_{i,j}^{cell}\,x_{i,j}.
\label{eq:milp-objective}
\end{equation}
where $U_i^{\text{rec}}$ represents the utility of record $i$ based on its visual salience, $U_{i,j}^{\text{cell}}$ represents the cell utility which measures the information retention, and $\alpha\in[0,1]$ is a parameter that balances the importance of these two parts. Both utility functions are further described below.

\paragraph{Record utility.}
Visual salience is quantified using the accurate pixel footprint $pc_i$ of record $i$ in the rendered tile. We normalize it by the maximum pixel footprint in the tile,
\begin{equation}
\widetilde{pc}_i = \frac{pc_i}{\max_j pc_j},
\end{equation}
and define the record utility as
\begin{equation}
U_i^{rec}
=
\lambda_{rec}\,\widetilde{pc}_i^{\,p},
\label{eq:record-utility}
\end{equation}
where $\lambda_{rec} > 0$ controls the strength of the salience preference and $p \ge 1$ controls its nonlinearity. The exponent $p$ amplifies differences in normalized pixel footprint, increasing the contrast between visually large and small records. This encourages the LP solver to retain large geometries, which would otherwise be disproportionately attractive to drop due to their higher byte cost when satisfying the size constraint.

\paragraph{Cell utility.}
For each attribute cell $(i,j)$, we measure the importance of retaining its value using a normalized divergence score $\widetilde{D}_{i,j} \in [0,1]$, e.g., normalized KL-divergence with respect to the reference distribution for attribute $j$ in $T_{in}$.
\begin{equation}
D_{i,j}
\;=\;
\mathrm{KLD}\!\left(\tilde p_j^{T}\,\middle\|\,\tilde p_j^{T^{(i,j\leftarrow\bot)}}\right).
\label{eq:kl-column}
\end{equation}

\begin{equation}
\widetilde{D}_{i,j}
\;=\;
\frac{D_{i,j}}{D^{\max}_j},
\qquad
D^{\max}_j
\;=\;
\max_{r\in\{1,\dots,N\}}
D_{r,j}.
\label{eq:kl-column-norm}
\end{equation}
Here, $\tilde p_j^{T^{(i,j\leftarrow\bot)}}$ denotes the empirical distribution of column $j$ after replacing the value in cell $(i,j)$ with $\bot$ (null) in table $T$, while keeping all other cells unchanged.

We define a \emph{keep-bonus}
\begin{equation}
\mathrm{KB}_{i,j} = 1 - \widetilde{D}_{i,j},
\label{eq:keep-bonus}
\end{equation}
and assign the cell utility
\begin{equation}
U_{i,j}^{cell}=\mathrm{KB}_{i,j}.
\label{eq:cell-utility}
\end{equation}

This formulation cleanly separates responsibilities where record utilities prioritize visually salient geometries, while cell utilities preserve the most informative attribute values within the remaining size budget. The MILP therefore removes small or visually insignificant records first and spends the remaining budget on attribute values that minimize distributional distortion.

Once the problem is formulated, we pass it to an LP solver to assign values to the decision variables. To create the output tile $T^{\text{out}}$, we make a copy of the input tile $T^{\text{in}}$ and set each cell at $f_i[j]$ to null if the corresponding variable $x_{ij}=0$. Additionally if $y_i=0$, we remove that entire row and if $u_j=0$ we remove the entire column.

\begin{algorithm}[t]
\small
\caption{\textsc{CellSparsifyMILP}($T_{\text{in}}, B$)}
\label{alg:cell-sparsify}
\begin{algorithmic}[1]
\Require Input tile $T_{\text{in}}=\{f_i\}_{i=1}^{N}$ with $d$ attributes (column 1 is geometry)
\Require Size threshold $B$
\Ensure Reduced tile $T_{\text{out}}$

\Statex \textbf{Compute record salience} \label{alg:cell-sparsify:salience}
\For{$i \gets 1$ to $N$}
    \State Compute pixel footprint $pc_i$ of record $i$ \label{alg:cell-sparsify:pc}
\EndFor
\For{$i \gets 1$ to $N$}
    \State Compute record utility $U_i^{rec}$ from $pc_i$ \label{alg:cell-sparsify:rec-util}
\EndFor

\Statex \textbf{Compute column statistics for size estimation} \label{alg:cell-sparsify:dict}
\For{$j \gets 2$ to $d$}
    \State Compute dictionary bytes $\widehat{b}^{dict}_j$ and non-null count $n_j$ \label{alg:cell-sparsify:dictbytes}
\EndFor
\For{$i \gets 1$ to $N$}
    \State Estimate geometry bytes $\widehat{b}^{geom}_i$ \label{alg:cell-sparsify:geombytes}
\EndFor

\Statex \textbf{Compute cell utilities} \label{alg:cell-sparsify:cell-util}
\For{each non-null cell $(i,j)$}
    \State Compute normalized divergence $\widetilde{D}_{i,j}$ \label{alg:cell-sparsify:kld}
    \State $U^{cell}_{i,j} \gets 1 - \widetilde{D}_{i,j}$ \label{alg:cell-sparsify:keepbonus}
\EndFor

\Statex \textbf{MILP formulation} \label{alg:cell-sparsify:milp}
\State Create MILP model $\mathcal{M}$
\State Add decision variables $y_i$, $u_j$, and $x_{i,j}$
\State Add objective function (Eq.~\ref{eq:milp-objective})
\State Add structural constraints linking $x_{i,j}$ to $y_i$ and $u_j$
\State Add linear size constraint (Eq.~\ref{eq:milp-size-constraint})

\State Solve $\mathcal{M}$ \label{alg:cell-sparsify:solve}

\State Construct $T_{\text{out}}$ from $(y,u,x)$ \label{alg:cell-sparsify:construct}
\State \Return $T_{\text{out}}$
\end{algorithmic}
\end{algorithm}

\autoref{alg:cell-sparsify} outlines the MILP formulation used to sparsify a vector tile.

Lines~\ref{alg:cell-sparsify:salience}--\ref{alg:cell-sparsify:rec-util} compute the visual salience of each record.
Line~\ref{alg:cell-sparsify:pc} measures the pixel footprint $pc_i$ of each geometry under the rendering function $\mathcal{R}$ (Eq.~\ref{eqn:render}).
These values are normalized and converted into record utilities $U_i^{rec}$ according to Eq.~\ref{eq:record-utility}, favoring geometries that occupy a larger visible area.

Lines~\ref{alg:cell-sparsify:dict}--\ref{alg:cell-sparsify:geombytes} compute statistics required by the linear size model.
For each attribute column $j$, Line~\ref{alg:cell-sparsify:dictbytes} estimates the dictionary size $\widehat{b}^{dict}_j$ and the number of non-null cells $n_j$.
The dictionary size corresponds to the bytes needed to encode the distinct non-null values of column $j$ in the shared vector-tile dictionary.
Line~\ref{alg:cell-sparsify:geombytes} estimates the geometry cost $\widehat{b}^{geom}_i$ of each record using its vertex count.
Lines~\ref{alg:cell-sparsify:cell-util}--\ref{alg:cell-sparsify:keepbonus} assign utilities to individual attribute cells.
For each non-null cell $(i,j)$, Line~\ref{alg:cell-sparsify:kld} computes the normalized divergence $\widetilde{D}_{i,j}$ using the column-wise KL-divergence defined in Eqs.~\ref{eq:kl-column}--\ref{eq:kl-column-norm}.
Line~\ref{alg:cell-sparsify:keepbonus} converts this score into a keep-bonus $\mathrm{KB}_{i,j}=1-\widetilde{D}_{i,j}$.
As a result, cells whose removal causes \emph{minimal} distributional change receive \emph{higher} utility and are more likely to be retained.
Lines~\ref{alg:cell-sparsify:milp}--\ref{alg:cell-sparsify:solve} formulate and solve the MILP.
The objective function (Eq.~\ref{eq:milp-objective}) balances record utilities and cell utilities using parameter $\alpha$, while the size constraint (Eq.~\ref{eq:milp-size-constraint}) ensures that the serialized tile does not exceed the target size $B$.
Finally, Line~\ref{alg:cell-sparsify:construct} constructs the output tile by removing records with $y_i=0$, dropping columns with $u_j=0$, and setting individual cells to null when $x_{i,j}=0$.

\section{Visualization-Aware Triage}
\label{sec:triage}

The sparsification problem and its MILP formulation perform well at reducing the tile size while preserving visualization fidelity. However, due to its complex formulation, it might not scale to extremely large tiles. This is especially needed for the tiles in the upper zoom levels that cover a very large geographical area, e.g., an entire country or continent. At an extreme case, the tile at level zero contains the entire dataset.

To address this problem, this section introduction the \emph{visualization-aware triage} problem which uses a set of heuristics to significantly reduce the tile size with little overhead to make it ready for the sparsification step.
The heuristics are grouped into two main approaches,
\emph{record-level triage} which makes horizontal cuts by keeping only high priority records, and
\emph{column-based triage} which removes some columns and reduces the size of some others.

\subsection{Record-Level Triage}
\label{sec:triage:record}
One natural strategy to limit the tile size is by reducing the number of records. A straight-forward solution is to pick a random sample with a fixed-size, e.g., reservoir sampling. The visualization-aware sampler~\cite{vas} focuses on unstyled scatter plots on points but is not applicable for non-point geometries or styling based on other attributes.

This part extends the idea of reservoir sampling by adding a priority that aims at preserving visualization fidelity by prioritizing records that are expected to contribute more to the visualization. This entails two questions that we need to address, first, how to define the reservoir capacity, and second, which priority to use when selecting records.


To address the first question, we have to decide on how to set the capacity of the tile which determines the constraint under which records are admitted into the reservoir. Notice that the reservoir capacity is different and typically much larger than the reduction threshold $B$ used by the sparsification algorithm. We considered four capacity metrics, each reflecting a different aspect of tile complexity:

\textbf{Number of Records:} This is a straight-forward and widely adopted option used in sampling algorithms such as traditional reservoir sampling or visualization-aware sampling (VAS)~\cite{vas}.

\textbf{Total Byte Size:} Directly controls the expected output size and is most closely aligned with loading time at the client.
 
 \textbf{Number of Vertices:} Captures geometric complexity and can directly impact the visualization quality and rendering time.

\textbf{Number of Cells:} Reflects the contribution of non-spatial attributes to overall tile size.

The second question to address is how to prioritize the records that are added to the tile. Sampling priorities determine which records are admitted first when competing for reservoir capacity which includes:

\textbf{Record Byte Size:} Used to choose between many small records or a fewer large records that contain more information.

\textbf{Number of Vertices:} Explores trade-offs between preserving fewer complex geometries versus many simple ones.
 \textbf{Tuple Size (Attributes):} Tests whether retaining many sparse records or fewer rich records yields higher quality.

 \textbf{Pixel Size:} Prioritizes features by projected screen footprint.

\textbf{Random:} Serves as an unbiased baseline.

Notice that the decisions of the reservoir capacity and the priority are orthogonal. For example, if the capacity is the \emph{total byte size} and the priority is \emph{largest pixel size}, we will logically sort all records by pixel size in descending order, and pick them one-by-one until the total byte size exceeds the capacity. Thus, we can use this approach to explore the different options and study their effect visualization fidelity.

Since our goal is to prepare the tile for the sparsification step, we decided to go with the \emph{number of cells} which directly impacts the processing cost of the LP solver. We let the sparsification step carry out the more comprehensive reduction as discussed earlier.
 
To better understand the relationships among candidate priority metrics, we computed the Spearman correlation matrices for three representative datasets (Provinces, Roads, and Countries). As shown in \autoref{tab:spearman}, we observe strong correlation between \emph{number of vertices}, \emph{number of pixels}, and \emph{byte size}. This makes sense because a geometry with a large number of vertices tends to have a larger area and takes more bytes. This means that ordering by any of them would result in roughly a similar order so we decided to adopt the \emph{largest number of vertices} as the default priority metric since it is the most straightforward to compute.


\begin{table}[t]
\centering
\footnotesize
\caption{Spearman correlation among priority metrics.}
\label{tab:spearman}
\setlength{\tabcolsep}{4pt} 
\begin{tabular}{lccc}
\toprule
Metric & Vertices & Pixels & Attributes \\
\midrule
\multicolumn{4}{l}{\textbf{Provinces}} \\
Pixels      & 0.743 & -- & -- \\
Attributes  & 0.195 & 0.089 & --  \\
Byte Size   & 0.979 & 0.731 & 0.226  \\
\midrule
\multicolumn{4}{l}{\textbf{Roads}} \\
Pixels      & 0.901 & -- & --  \\
Attributes  & 0.061 & 0.070 & -- \\
Byte Size   & 0.879 & 0.811 & 0.460  \\
\bottomrule
\end{tabular}
\end{table}

To summarize, the record-level triage offers a computationally inexpensive and highly parallel way to reduce large tiles before any fine-grained analysis occurs. However, because it discards entire records, it may remove visually or semantically important content, even when only specific attributes or cells within those records contribute to the oversize problem. As a result, pure record-level triage offers efficiency but may not always preserve fidelity in dense regions.

\subsection{Column-Level Triage}
\label{sec:triage:column}

The second set of heuristics work at the column level and aim to reduce the number of columns and the amount of information stored in them. This step is performed before sparsification, if the tile exceeds size constraint after record-level triage.

\subsubsection{Attribute Subset Selection}

At first, this step reduces the size by removing some columns not used in visualization. Any column that has the same value in all rows, i.e., zero entropy, can be safely removed since it does not result in any visual information. Further, users can manually, or with the assist of LLM~\cite{10.14778/3750601.3750690}, remove any attributes that are not relevant to the application in hand. For example, some datasets contain several internal identifiers that are not meaningful to the user or contains labels in various languages and some of them can be removed. This step is optional and also users may choose to proceed with the entire datasets. 

\subsubsection{Geometry Simplification}
Since the size of the tile at each level changes, this gives us the opportunity to simplify some geometries to reduce their level of detail and size while preserving their visual fidelity at the target zoom level. For this step, the system offers three approaches: Douglas-Peucker, Topology-Preserving, or grid snapping simplification.

Douglas-Peucker takes a geometry and a tolerance distance $\varepsilon$. It then iteratively connects the two endpoints and removes any intermediate points that are closer than $\varepsilon$. The geometry is then split into two parts, and the process is repeated. By setting $\varepsilon$ closer to the pixel size at the corresponding zoom level, the simplified geometry will look visually similar to the original one.

The topology-preserving simplifier works similar to Douglas-Peucker but adds additional constraints to ensure that the simplified geometry does not cause any self intersection or change topology-changing issues.

The grid-snapping simplification applies a uniform grid of a fixed size and snaps each vertex to the closest grid cell. It then removes multiple consecutive points that snap to the same grid cell.

In addition to simplification, the system implements geometry clipping to ensure that only portions of geometries that intersect with the tile boundaries are kept. Geometries outside the tile's bounds are removed, while those crossing boundaries are trimmed to the tile's extent. The system handles special cases for different geometry types, e.g., points, multipoints, linestrings, and polygons, applying appropriate simplification strategies to each.

For the final rendering, all coordinates are converted to integer coordinates according to the Mapbox Vector Tile (MVT) specification.


\subsubsection{Numeric Quantization}
For numeric columns, we apply a quantization step that groups values into a few discreet values, e.g., 10 by default. Given that these values are typically used to create a coloring scheme for records, the difference in color can be tolerated in visualization, especially in color scales. Column-based and dictionary-based stores, e.g., MVT, use less disk space after quantization. Furthermore, quantization reduces the entropy of the column which allows the sparsification step to give these attributes higher priority. To reduce computation time, only numeric columns with more than 20 unique values are considered for quantization.

\subsubsection{String Trimming}
To reduce the size of text attributes, we trim them to a fixed length. String attributes are typically used as labels and trimming them would significantly reduce the tile size while still giving a hint to the user on what the label is. For example, if the original label is ``Massachusets'', shortening it to ``Mass$\dots$'', would still be informative. To measure the information loss, we compute the KL-divergence between the original values and the trimmed ones while mapping the values correspondingly. If the trimming causes labels to merge together, e.g., ``Santa Cruz'' and ``Santa Clara'' becoming both ``Santa C$\dots$'', then it will result in a higher divergence score.

\subsubsection{Prioritized Column Reduction}
For both \emph{numeric quantization} and \emph{string trimming}, we prioritize attribute reduction by employgin the KL divergence between each attribute's original and post-reduction probability distributions. This information-theoretic measure quantifies the expected information loss resulting from the transformation. Lower KL divergence values indicate that the reduction would preserve more of the attribute's original information content.
The reduction process proceeds iteratively, targeting attributes with the smallest change in KL divergence first. After each attribute reduction, the system recalculates the tile size and continues the process until either the size falls below the target threshold or all attributes have been processed.

\begin{figure*}[t]
    \centering

    \begin{subfigure}[t]{0.16\textwidth}
        \centering
        \includegraphics[width=\linewidth]{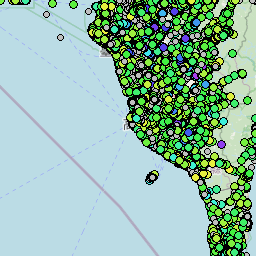}
        \caption{eBird (Grad)}
    \end{subfigure}\hfill
    \begin{subfigure}[t]{0.16\textwidth}
        \centering
        \includegraphics[width=\linewidth]{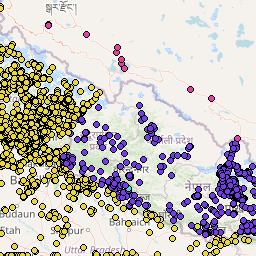}
        \caption{eBird (Cat)}
    \end{subfigure}\hfill
    \begin{subfigure}[t]{0.16\textwidth}
        \centering
        \includegraphics[width=\linewidth]{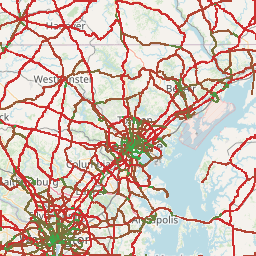}
        \caption{Roads (Grad)}
    \end{subfigure}\hfill
    \begin{subfigure}[t]{0.16\textwidth}
        \centering
        \includegraphics[width=\linewidth]{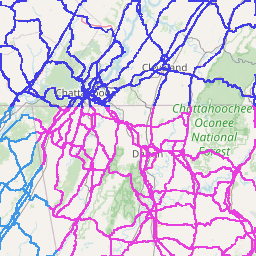}
        \caption{Roads (Cat)}
    \end{subfigure}\hfill
    \begin{subfigure}[t]{0.16\textwidth}
        \centering
        \includegraphics[width=\linewidth]{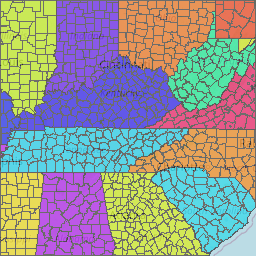}
        \caption{Counties (Cat)}
    \end{subfigure}\hfill
    \begin{subfigure}[t]{0.16\textwidth}
        \centering
        \includegraphics[width=\linewidth]{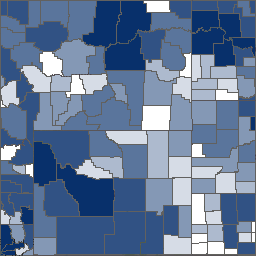}
        \caption{Counties (Grad)}
    \end{subfigure}

    \caption{
    Examples of visualization stylings used in our experiments.
    We show two stylings for \emph{eBird} (gradient based on observation time of the day and categorical based on country code),
    two for \emph{Roads} (gradient based on length and categorical based on state), and two for \emph{Counties} (gradient based on area of water and categorical based on state).
    }
    \label{fig:all_stylings_overview}
\end{figure*}

\section{Experiments} \label{sec:experiments}

In this section, we run a comprehensive evaluation on \hifive to provide experimental evidence of its scalability. We also evaluate the stylized vector tiles produced using several objective metrics to evaluate the visualization quality.

\subsection{Experimental Setup}
We ran experiments on a twelve‐node Spark cluster. The master has two 8‐core Intel Xeon E5-2609 v4 CPUs (1.70 GHz) and 128 GB RAM, while each worker node has two 6‐core Intel Xeon E5-2603 v4 CPUs (1.70 GHz) and 64 GB RAM. All machines run CentOS 7.5.1804, using local SSDs for the OS and HDDs for data storage. 

\begin{table}[t]
\centering
\footnotesize
\caption{Main datasets used in our experiments.}
\label{tab:datasets}
\begin{tabular}{lrrl}
\toprule
\textbf{Dataset} & \textbf{Size} & \textbf{\#Records} & \textbf{Geometry} \\
\midrule
eBird~\cite{ebird}                & 935.3 GB & 801M & Points \\
OSM\_Roads~\cite{EM15}            & 31 GB    & 70M  & Linestrings \\
OSM\_Buildings~\cite{EM15}        & 187 GB   & 455M & Polygons \\
Postal\_Codes~\cite{EM15}         & 7 GB     & 236k & Polygons \\
Roads\_NA~\cite{UCRSTAR/NE/roads_NA} & 81 MB & 49k  & Linestrings \\
Counties~\cite{UCRSTAR/TIGER2018/COUNTY} & 193 MB & 3k & Polygons \\
\bottomrule
\end{tabular}
\end{table}

\autoref{tab:datasets} lists the datasets used in experiments. The datasets includes points, lines, and polygons with up-to 935GB. For large-scale datasets such as eBird, OSM Roads, and OSM Buildings, we additionally generate multiple subsets using uniform random sampling at different ratios (e.g., 10\%, 1\%, 0.1\%, and smaller). These subsets enable scalability experiments and allow us to study the impact of tile-size constraints on both performance and visual fidelity across a wide range of data volumes, while preserving the spatial and attribute characteristics of the original datasets.

\textbf{Baselines:} We consider two baselines. AID*~\cite{aidstar} which is a Spark-based visualization index that combines pre-generated tiles with a data index to support scalable visualization of raster tiles. Tippecanoe~\cite{tippecanoe} is the tool used by Mapbox~\cite{mvt} to create reduced-size vector tiles.

\textbf{Parameters:} \autoref{tab:exp-settings} summarizes the experimental parameters that we vary and their corresponding default values.

\begin{table}[t]
\centering
\footnotesize
\caption{Experimental settings.}
\label{tab:exp-settings}
\begin{tabular}{ll}
\toprule
\textbf{Setting} & \textbf{Values (default)} \\
\midrule
Tile size constraint ($B$) & 32 KB -- 4 MB (256 KB) \\
Triage priority            & Highest number of vertices \\
Tile capacity              & (100K cells) \\
Zoom levels                & 1 -- 15 (8) \\
$\alpha$                   & 0.0 -- 1.0 (0.5) \\
Input sample size          & 0.001\% -- 100\% (100\%) \\ 
\bottomrule
\end{tabular}
\end{table}

\subsection{Quality Metrics}
\label{sec:quality-metrics}

To evaluate visual fidelity, we render each tile (baseline and reduced) into an RGB image of size
$H \times W$ (in our experiments $H=W=256$) using the same client-side styling.
\autoref{fig:all_stylings_overview} summarizes the styles we use in the evaluation.
Let $A,B \in \{0,\dots,255\}^{H\times W\times 3}$ denote the resulting baseline and reduced images,
and let $A_{u,v,c}$ be the intensity of channel $c\in\{1,2,3\}$ at pixel $(u,v)$.

\paragraph{RMSE}
We compute the per-pixel mean squared error (MSE) over all pixels and channels,
\begin{equation}
\mathrm{MSE}(A,B)
=
\frac{1}{3HW}\sum_{u=1}^{H}\sum_{v=1}^{W}\sum_{c=1}^{3}\left(A_{u,v,c}-B_{u,v,c}\right)^2,
\end{equation}
and report the root mean squared error (RMSE),
\begin{equation}
\mathrm{RMSE}(A,B)=\sqrt{\mathrm{MSE}(A,B)}.
\end{equation}

\paragraph{PSNR}
Peak signal-to-noise ratio (PSNR) is reported in decibels (dB) using the 8-bit peak value $I_{\max}=255$. Higher PSNR indicates closer visual agreement.
{\small
\begin{equation}
\mathrm{PSNR}(A,B)
=
\begin{cases}
+\infty & \text{if }\mathrm{MSE}(A,B)=0,\\[2pt]
20\log_{10}(I_{\max}) - 10\log_{10}(\mathrm{MSE}(A,B)) & \text{otherwise.}
\end{cases}
\end{equation}
}

\paragraph{Structural Similarity Index (SSIM)}
The SSIM index~\cite{WBS+04} is a perceptual metric designed to align with human visual judgment by comparing local patterns of luminance, contrast, and structure rather than global pixel-wise errors. The metric operates by sliding a fixed-size window across the reference image $A$ and compared image $B$, extracting local patches $x$ and $y$ at each location. The similarity between these patches is defined as:
\begin{equation}
\mathrm{SSIM}(x,y) = \frac{(2\mu_x\mu_y + C_1)(2\sigma_{xy}+C_2)}{(\mu_x^2+\mu_y^2+C_1)(\sigma_x^2+\sigma_y^2+C_2)}
\end{equation}
where $\mu$ and $\sigma$ denote the local mean and variance, $\sigma_{xy}$ is the covariance, and $C_1, C_2$ are stabilizing constants. 

The final score is the mean SSIM aggregated over all window locations and color channels. In our evaluation, we use $7 \times 7$ windows on an intensity range of $[0,255]$. SSIM values range from $[-1,1]$, with higher values indicating superior structural fidelity. Future work could investigate adapting this optimization objective to directly minimize the tile-level distortion defined in our problem formulation, potentially bypassing the need for full image rendering during the reduction process.

\subsection{End-to-end Evaluation}

\begin{figure}[t]
\centering

\pgfplotstableread{
Records    Hi5_Time AID_Time Tippe_Time Hi5_Size AID_Size Tippe_Size
8000       95.0     39.0     10.0       16.3     0.003    1.7
80000      78.0     42.0     60.0       69.9     0.03     4.5
800000     96.0     79.0     300.0      239.3    0.3      20.0
8000000    184.0    114.0    3060.0     543.9    1.4      49.0
80000000   820.0    223.0    nan        971.3    27.6     nan
801000000  3420.0   1560.0   nan        926.7    28.0     nan
}\scalingDataEbird

\pgfplotstableread{
Records    Hi5_Time AID_Time Tippe_Time Hi5_Size AID_Size Tippe_Size
7000       88.0     34       6          4.6      1.6      4.9
70000      107.0    39       17         30.4     4.1      29
700000     121.0    57       141        155.7    8.7      150
7000000    219.0    88       1623       654.1    23.2     601
70000000   913.0    158      nan        1650     50.6     nan
}\scalingDataRoads

\pgfplotstableread{
Records    Hi5_Time AID_Time Tippe_Time Hi5_Size AID_Size Tippe_Size
45000      87.0     34       6          0.2      1.4      4.9
450000     77.0     74       142        1.5      3.2      12.2
4500000    174.0    87       594        10.4     6.3      123
45000000   439.0    142      nan        56.2     12.2     nan
450000000  1969.0   759      nan        255.5    19.2     nan
}\scalingDataBuildings

\begin{tikzpicture}[scale=0.9]
\begin{groupplot}[
    group style={
        group size=2 by 3,
        horizontal sep=0.85cm,  
        vertical sep=1.0cm,     
        xlabels at=edge bottom,
        xticklabels at=edge bottom,
        ylabels at=edge left,
        yticklabels at=edge left
    },
    width=0.385\columnwidth,   
    height=0.32\columnwidth,   
    scale only axis,
    xmode=log,
    ymode=log,
    grid=major,
    tick label style={font=\normalsize},
    label style={font=\small},
    ylabel style={yshift=-2pt},
    xlabel style={yshift=0pt},
    title style={font=\normalsize, yshift=-1ex}
]

\nextgroupplot[
    ylabel={Runtime (s) (Log)},
    title={Computation},
]
    \addplot[color=red,   line width=0.8pt, mark=square*, mark size=1.2pt]
        table[x=Records, y=Hi5_Time] {\scalingDataEbird};
    \addplot[color=blue,  mark=x, mark size=1.5pt]
        table[x=Records, y=AID_Time] {\scalingDataEbird};
    \addplot[color=black, mark=o, mark size=1.5pt]
        table[x=Records, y=Tippe_Time] {\scalingDataEbird};

\nextgroupplot[
    ylabel={Output Size (MB) (Log)},
    title={Storage},
]
    \addplot[color=red,   line width=0.8pt, mark=square*, mark size=1.2pt]
        table[x=Records, y=Hi5_Size] {\scalingDataEbird};
    \addplot[color=blue,  mark=x, mark size=1.5pt]
        table[x=Records, y=AID_Size] {\scalingDataEbird};
    \addplot[color=black, mark=o, mark size=1.5pt]
        table[x=Records, y=Tippe_Size] {\scalingDataEbird};

\nextgroupplot[
    ylabel={Runtime (s) (Log)},
]
    \addplot[color=red,   line width=0.8pt, mark=square*, mark size=1.2pt]
        table[x=Records, y=Hi5_Time] {\scalingDataRoads};
    \addplot[color=blue,  mark=x, mark size=1.5pt]
        table[x=Records, y=AID_Time] {\scalingDataRoads};
    \addplot[color=black, mark=o, mark size=1.5pt]
        table[x=Records, y=Tippe_Time] {\scalingDataRoads};

\nextgroupplot[
    ylabel={Output Size (MB) (Log)},
]
    \addplot[color=red,   line width=0.8pt, mark=square*, mark size=1.2pt]
        table[x=Records, y=Hi5_Size] {\scalingDataRoads};
    \addplot[color=blue,  mark=x, mark size=1.5pt]
        table[x=Records, y=AID_Size] {\scalingDataRoads};
    \addplot[color=black, mark=o, mark size=1.5pt]
        table[x=Records, y=Tippe_Size] {\scalingDataRoads};

\nextgroupplot[
    ylabel={Runtime (s) (Log)},
    xlabel={Number of Records (Log)}
]
    \addplot[color=red,   line width=0.8pt, mark=square*, mark size=1.2pt]
        table[x=Records, y=Hi5_Time] {\scalingDataBuildings};
    \addplot[color=blue,  mark=x, mark size=1.5pt]
        table[x=Records, y=AID_Time] {\scalingDataBuildings};
    \addplot[color=black, mark=o, mark size=1.5pt]
        table[x=Records, y=Tippe_Time] {\scalingDataBuildings};

\nextgroupplot[
    ylabel={Output Size (MB) (Log)},
    xlabel={Number of Records (Log)}
]
    \addplot[color=red,   line width=0.8pt, mark=square*, mark size=1.2pt]
        table[x=Records, y=Hi5_Size] {\scalingDataBuildings};
    \addplot[color=blue,  mark=x, mark size=1.5pt]
        table[x=Records, y=AID_Size] {\scalingDataBuildings};
    \addplot[color=black, mark=o, mark size=1.5pt]
        table[x=Records, y=Tippe_Size] {\scalingDataBuildings};

\end{groupplot}

\matrix[
  matrix of nodes,
  anchor=north,
  nodes={font=\small, inner sep=1pt},
  column sep=6pt
] at ($(group c1r3.south)!0.5!(group c2r3.south) - (0,0.80cm)$) { 
  {\tikz{\draw[red, line width=0.8pt] (0,0) -- (0.3,0);
         \fill[red] (0.12,-0.03) rectangle +(0.06,0.06);}}; & \node{Hi5}; &
  {\tikz{\draw[blue] (0,0) -- (0.3,0);
         \draw[blue] (0.12,-0.03) -- (0.18,0.03)
                      (0.12,0.03) -- (0.18,-0.03);}}; & \node{AID*}; &
  {\tikz{\draw[black] (0,0) -- (0.3,0);
         \draw[black] (0.15,0) circle (0.04);}}; & \node{Tippecanoe}; \\
};

\node[anchor=east, font=\small] at ($(group c1r1.west) + (-0.85cm,0)$) {eBird};
\node[anchor=east, font=\small] at ($(group c1r2.west) + (-0.85cm,0)$) {Roads};
\node[anchor=east, font=\small] at ($(group c1r3.west) + (-0.85cm,0)$) {Buildings};

\end{tikzpicture}
\caption{Scaling comparison across eBird subsets (8k to 801M records), roads subsets (7k to 70M records), and buildings subsets (45k to 455M records). Both axes are logarithmic.}
\label{fig:scaling_comparison}
\end{figure}
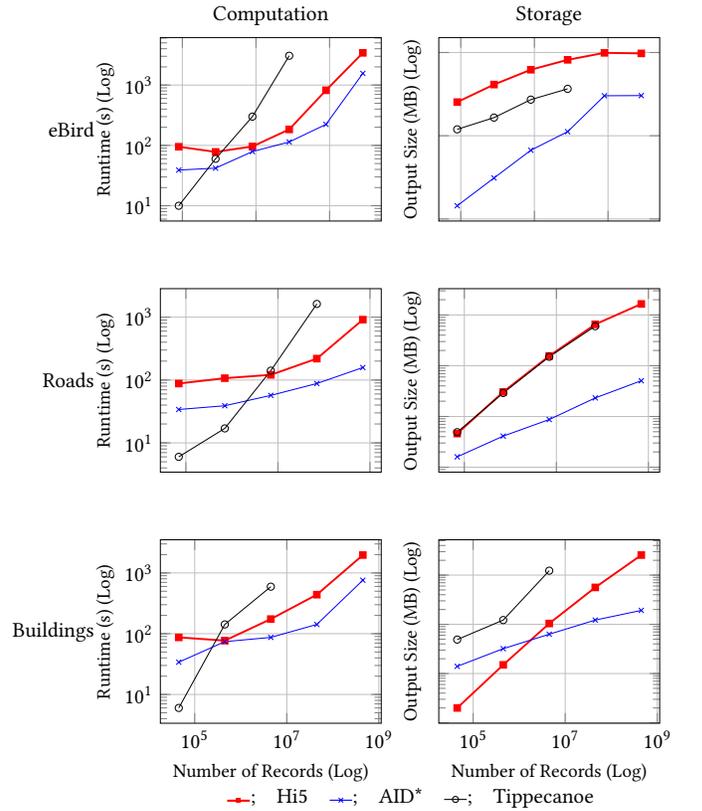
\mycomment{
\begin{figure}[t]
\centering
\pgfplotstableread{
Records    Hi5_Time AID_Time Tippe_Time Hi5_Size AID_Size Tippe_Size
7000       88.0    34       6          4.6      1.6       4.9
70000      107.0    39       17         30.4     4.1      29
700000     121.0    57       141        155.7    8.7       150
7000000    219.0    88       1623       654.1    23.2      601
70000000   913.0    158      nan        1650     50.6      nan
}\scalingDataRoads

\begin{tikzpicture}[scale=0.9]
\begin{groupplot}[
    group style={
        group size=2 by 1,
        horizontal sep=1.1cm,
        xlabels at=edge bottom,
        xticklabels at=edge bottom
    },
    width=0.41\columnwidth,         
    height=0.35\columnwidth,
    scale only axis,
    xmode=log,
    ymode=log,
    grid=major,
    tick label style={font=\normalsize},
    label style={font=\small},
    title style={font=\normalsize, yshift=-1ex}
]

\nextgroupplot[
    ylabel={Runtime (s) (Log)},
    title={Computation},
    xlabel={Number of Records (Log)}
]
    \addplot[color=red, line width=0.8pt, mark=square*, mark size=1.2pt]
        table[x=Records, y=Hi5_Time] {\scalingDataRoads};
    
    \addplot[color=blue, mark=x, mark size=1.5pt]
        table[x=Records, y=AID_Time] {\scalingDataRoads};
    
    \addplot[color=black, mark=o, mark size=1.5pt]
        table[x=Records, y=Tippe_Time] {\scalingDataRoads};

\nextgroupplot[
    ylabel={Output Size (MB) (Log)},
    title={Storage},
    xlabel={Number of Records (Log)}
]
    \addplot[color=red, line width=0.8pt, mark=square*, mark size=1.2pt]
        table[x=Records, y=Hi5_Size] {\scalingDataRoads};
    
    \addplot[color=blue, mark=x, mark size=1.5pt]
        table[x=Records, y=AID_Size] {\scalingDataRoads};
    
    \addplot[color=black, mark=o, mark size=1.5pt]
        table[x=Records, y=Tippe_Size] {\scalingDataRoads};

\end{groupplot}

\matrix[
  matrix of nodes,
  anchor=north,
  nodes={font=\small, inner sep=1pt},
  column sep=6pt
] at ($(group c1r1.south)!0.5!(group c2r1.south) - (0,0.8cm)$) {
  {\tikz{\draw[red, line width=0.8pt] (0,0) -- (0.3,0); 
              \fill[red] (0.12,-0.03) rectangle +(0.06,0.06);}}; & \node{Hi5}; &
  
  {\tikz{\draw[blue] (0,0) -- (0.3,0); 
              \draw[blue] (0.12,-0.03) -- (0.18,0.03) (0.12,0.03) -- (0.18,-0.03);}}; & \node{AID*}; &
  
  {\tikz{\draw[black] (0,0) -- (0.3,0); 
              \draw[black] (0.15,0) circle (0.04);}}; & \node{Tippecanoe}; \\
};
\end{tikzpicture}
\caption{Scaling comparison across roads subsets (7k to 70M records). Both axes are logarithmic.}
\label{fig:scaling_comparison-roads}
\end{figure}

\begin{figure}[t]
\centering
\pgfplotstableread{
Records    Hi5_Time AID_Time Tippe_Time Hi5_Size AID_Size Tippe_Size
45000       87.0    34       6          0.2      1.4       4.9
450000      77.0    74       142        1.5      3.2      12.2
4500000     174.0   87       594        10.4     6.3      123
45000000    439.0   142      nan        56.2     12.2     nan
450000000   1969.0  759      nan        255.5    19.2     nan
}\scalingDataBuildings

\begin{tikzpicture}[scale=0.9]
\begin{groupplot}[
    group style={
        group size=2 by 1,
        horizontal sep=1.1cm,
        xlabels at=edge bottom,
        xticklabels at=edge bottom
    },
    width=0.41\columnwidth,
    height=0.35\columnwidth,
    scale only axis,
    xmode=log,
    ymode=log,
    grid=major,
    tick label style={font=\normalsize},
    label style={font=\small},
    title style={font=\normalsize, yshift=-1ex}
]

\nextgroupplot[
    ylabel={Runtime (s) (Log)},
    title={Computation},
    xlabel={Number of Records (Log)}
]
    \addplot[color=red, line width=0.8pt, mark=square*, mark size=1.2pt]
        table[x=Records, y=Hi5_Time] {\scalingDataBuildings};

    \addplot[color=blue, mark=x, mark size=1.5pt]
        table[x=Records, y=AID_Time] {\scalingDataBuildings};

    \addplot[color=black, mark=o, mark size=1.5pt]
        table[x=Records, y=Tippe_Time] {\scalingDataBuildings};

\nextgroupplot[
    ylabel={Output Size (MB) (Log)},
    ylabel style={yshift=-6pt}, 
    title={Storage},
    xlabel={Number of Records (Log)}
]
    \addplot[color=red, line width=0.8pt, mark=square*, mark size=1.2pt]
        table[x=Records, y=Hi5_Size] {\scalingDataBuildings};

    \addplot[color=blue, mark=x, mark size=1.5pt]
        table[x=Records, y=AID_Size] {\scalingDataBuildings};

    \addplot[color=black, mark=o, mark size=1.5pt]
        table[x=Records, y=Tippe_Size] {\scalingDataBuildings};

\end{groupplot}

\matrix[
  matrix of nodes,
  anchor=north,
  nodes={font=\small, inner sep=1pt},
  column sep=6pt
] at ($(group c1r1.south)!0.5!(group c2r1.south) - (0,0.8cm)$) {

  {\tikz{\draw[red, line width=0.8pt] (0,0) -- (0.3,0);
         \fill[red] (0.12,-0.03) rectangle +(0.06,0.06);}}; & \node{Hi5}; &

  {\tikz{\draw[blue] (0,0) -- (0.3,0);
         \draw[blue] (0.12,-0.03) -- (0.18,0.03)
                      (0.12,0.03) -- (0.18,-0.03);}}; & \node{AID*}; &

  {\tikz{\draw[black] (0,0) -- (0.3,0);
         \draw[black] (0.15,0) circle (0.04);}}; & \node{Tippecanoe}; \\
};
\end{tikzpicture}
\caption{Scaling comparison across buildings subsets (45k to 455M records). Both axes are logarithmic.}
\label{fig:scaling_comparison-buildings}
\end{figure}
}
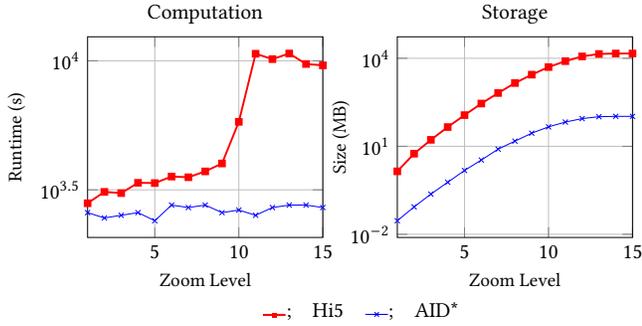
\begin{figure}[t]
\centering
\pgfplotstableread{
Zoom LP_Time AID_Time LP_Size_MB AID_Size_MB
1    2805.0  2580.0   1.4         0.029
2    3105.0  2460.0   5.6         0.086
3    3071.0  2520.0   16.6        0.235
4    3367.0  2580.0   45.3        0.595
5    3359.0  2400.0   115.6       1.500
6    3560.0  2760.0   286.7       3.400
7    3538.0  2700.0   654.8       8.000
8    3728.0  2760.0   1433.0      15.000
9    4001.0  2580.0   2764.6      28.000
10   5801.0  2640.0   5017.6      46.000
11   10653.0 2520.0   7987.6      67.000
12   10153.0 2700.0   11571.0     88.000
13   10675.0 2760.0   14028.4     102.000
14   9720.0  2760.0   14540.2     105.000
15   9605.0  2700.0   14540.2      105.000
}\perfData

\begin{tikzpicture}[scale=0.9]
\begin{groupplot}[
    group style={
        group size=2 by 1,
        horizontal sep=1.1cm,
        xlabels at=edge bottom,
        xticklabels at=edge bottom
    },
    width=0.41\columnwidth,
    height=0.35\columnwidth,
    scale only axis,
    xmin=1, xmax=15,
    ymode=log,
    grid=major,
    tick label style={font=\normalsize},
    label style={font=\small},
    title style={font=\normalsize, yshift=-1ex}
]

\nextgroupplot[
    ylabel={Runtime (s)},
    title={Computation},
    xlabel={Zoom Level}
]
    \addplot[color=red, line width=0.8pt, mark=square*, mark size=1.2pt]
        table[x=Zoom, y=LP_Time] {\perfData};

    \addplot[color=blue, mark=x, mark size=1.5pt]
        table[x=Zoom, y=AID_Time] {\perfData};

\nextgroupplot[
    ylabel={Size (MB)},
    ylabel style={yshift=-6pt}, 
    title={Storage},
    xlabel={Zoom Level}
]
    \addplot[color=red, line width=0.8pt, mark=square*, mark size=1.2pt]
        table[x=Zoom, y=LP_Size_MB] {\perfData};

    \addplot[color=blue, mark=x, mark size=1.5pt]
        table[x=Zoom, y=AID_Size_MB] {\perfData};

\end{groupplot}

\matrix[
  matrix of nodes,
  anchor=north,
  nodes={font=\small, inner sep=1pt},
  column sep=6pt
] at ($(group c1r1.south)!0.5!(group c2r1.south) - (0,0.8cm)$) {

  {\tikz{\draw[red, line width=0.8pt] (0,0) -- (0.3,0);
         \fill[red] (0.12,-0.03) rectangle +(0.06,0.06);}}; & \node{Hi5}; &

  {\tikz{\draw[blue] (0,0) -- (0.3,0);
         \draw[blue] (0.12,-0.03) -- (0.18,0.03)
                      (0.12,0.03) -- (0.18,-0.03);}}; & \node{AID*}; \\
};
\end{tikzpicture}
\caption{Performance comparison across zoom levels 1--15 for eBird dataset. Runtime and output size are shown on a logarithmic scale for the LP-based method (Hi5) and AID*.}
\label{fig:perf-zoom-singlecol}
\end{figure}

This section evaluates end-to-end tile generation running time, storage footprint, and visualization fidelity for \hifive against Tippecanoe and AID*.
Figure \ref{fig:scaling_comparison} 
shows runtimes and storage footprint of the three largest datasets, eBird, OSM\_Roads, and OSM\_Buildings, while scaling input size from 0.001\% to 100\%. For each setting, we generate a pyramid of eight zoom levels (close to 20,000 tiles in total) and report the total output size. Since \hifive and AID* use Spark, both scale to the largest dataset. Notably, \hifive processes nearly a terabyte dataset in less than an hour. The resulting tiles can be used to generate any number of styles in real-time with no overhead on the server.
In contrast, because Tippecanoe is optimized for a single machine, it shows good performance only for small datasets and runs out of memory for larger ones. Output size increases with dataset scale for all methods as more tiles become populated; raster tiles (AID*) remain substantially smaller because they encode only pixels rather than geometries and attributes. Finally, even under the same nominal tile budget, \hifive and Tippecanoe yield different realized sizes due to fundamentally different reduction strategies and size models.

\autoref{fig:perf-zoom-singlecol} highlights the effect of the pyramid depth varies from 1 to 15 on the runtime and storage size for the eBird dataset (Tippecanoe is omitted as it runs out of memory). As expected, the unstyled raster tiles of AID* exhibit the smallest size and fastest runtime at the expense of not allowing any client-side styling. The runtime and size of \hifive initially increase due to the large number of tiles but they eventually stabilize as the size of each tile is upper-bounded.

Figure 
~\ref{fig:roads-metrics} evaluates the visualization quality on stylized rendering for roads. Due to Tippecanoe's scalability issues, we run this experiment on a 1\% sample. Both \hifive and Tippecanoe are configured with the same target tile budget, $B=256\text{kb}$. The triage step in \hifive was configured with a capacity of 100k cells and prioritize geometries with more vertices. 
The coloring styles applied to the roads dataset were \emph{categorical} based on the ``highway" attribute, and \emph{gradient} based on the ``max speed'' attribute. We observed consistent trends across other datasets, but omit their results due to space limitations.
Across both styles, \hifive delivers the highest fidelity, i.e., highest SSIM and PSNR, and lowest RMSE. The difference is most prominent for SSIM which correlates more closely with perceived visual structure which indicates that the visualization-aware reduction in \hifive is better at preserving visualization characteristics. We notice that PSNR and RMSE show the same behavior and they can also be used to assess visualization quality to some limit. Finally, AID* is not competitive with these metrics because it creates raster tiles with no styling options for users.

\mycomment{
\begin{figure}[t]
\centering

\pgfplotstableread{
 z  mean_ssim  mean_rmse  mean_psnr
  0  0.780902  45.940070  14.886971
  1  0.816083  38.961508  17.183405
  2  0.772037  50.347769  14.761022
  3  0.857523  46.204386  16.687021
  4  0.889084  46.003724  16.662720
  5  0.846260  44.610703  19.151280
  6  0.898874  37.417179  21.873629
  7  0.898385  34.949677  23.640131
  8  0.966008  19.856926  28.629486
  9  0.997118   5.993157  35.666923
 10  0.999729   2.220182  41.054757
 11  0.999941   0.761526  46.493801
 12  0.999977   0.217519  51.601769
 13  0.999991   0.047245  54.836429
 14  0.999996   0.015005  53.153912
}\CatTableHiFiveEbird

\pgfplotstableread{
 z  mean_ssim  mean_rmse  mean_psnr
  0  0.021436  58.826355  12.739365
  1  0.038746  48.908505  14.897462
  2  0.023376  66.196149  12.299277
  3  0.046393  64.971532  13.447445
  4  0.044996  66.613964  12.831697
  5  0.033556  65.029741  12.674822
  6  0.035644  62.287166  13.246814
  7  0.039405  59.193752  13.653544
  8  0.047937  53.029272  14.554246
  9  0.074649  48.532693  15.690881
 10  0.093418  42.227277  17.114800
 11  0.118614  36.415519  18.544418
 12  0.156310  28.726783  20.470895
 13  0.171473  23.759407  21.427306
 14  0.262488  17.284003  24.366675
}\CatTableTippeEbird

\pgfplotstableread{
 z  mean_ssim  mean_rmse  mean_psnr
  0  0.000003  244.008517   0.382704
  1  0.000004  246.311273   0.304138
  2  0.000003  240.147785   0.533234
  3  0.000004  239.243949   0.568266
  4  0.000004  237.099847   0.653247
  5  0.000004  235.375825   0.727310
  6  0.000005  234.286693   0.767996
  7  0.000004  235.591626   0.711084
  8  0.000005  239.957419   0.541221
  9  0.000005  244.246289   0.381031
 10  0.000006  247.412105   0.265979
 11  0.000007  249.516406   0.190141
 12  0.000008  250.753876   0.146316
 13  0.000007  251.163501   0.131845
}\CatTableAidEbird

\pgfplotstableread{
 z  mean_ssim  mean_rmse  mean_psnr
  0  0.950918  15.631374  24.250860
  1  0.895844  29.002868  18.882828
  2  0.774641  37.718735  17.233246
  3  0.802148  42.127857  15.671348
  4  0.811646  43.632627  15.195943
  5  0.798224  44.125239  15.063249
  6  0.810920  43.004724  15.322724
  7  0.853732  34.537120  17.077312
  8  0.946634  15.341416  24.463512
  9  0.971891  10.013539  29.648691
 10  0.975014   8.984498  30.564721
 11  0.975897   8.655264  30.889646
 12  0.976720   8.289485  37.112174
 13  0.993134   4.116203  43.657766
 14  0.996754   2.275518  48.970482
}\GradTableHiFiveEbird

\pgfplotstableread{
 z  mean_ssim  mean_rmse  mean_psnr
  0  0.966523  13.022653  25.836814
  1  0.906819  28.196069  19.132453
  2  0.781872  37.387701  17.358699
  3  0.810800  42.296435  15.648225
  4  0.819173  43.976976  15.126610
  5  0.806193  44.544484  14.972836
  6  0.819096  43.101625  15.303195
  7  0.861304  34.567570  17.085279
  8  0.952718  15.642799  24.274251
  9  0.976508  10.170100  29.512715
 10  0.979476   9.087933  30.477505
 11  0.980101   8.724727  30.830616
 12  0.905001  18.668525  24.625141
 13  0.941983  13.520820  27.359196
 14  0.957727  11.265492  28.915681
}\GradTableTippeEbird

\pgfplotstableread{
 z  mean_ssim  mean_rmse  mean_psnr
  0  0.000097  254.274335   0.024753
  1  0.000092  251.818692   0.109051
  2  0.000079  247.820884   0.249610
  3  0.000082  247.581308   0.256559
  4  0.000073  241.320155   0.481674
  5  0.000072  240.237750   0.520594
  6  0.000073  240.944112   0.494979
  7  0.000078  244.758985   0.356399
  8  0.000091  251.333442   0.125786
  9  0.000096  253.825904   0.040202
 10  0.000094  253.557514   0.049388
 11  0.000093  253.358477   0.056181
 12  0.000090  252.546375   0.084510
 13  0.000094  253.760772   0.042400
 14  0.000095  254.069106   0.031819
}\GradTableAidEbird

\pgfplotsset{
  myCycle/.style={
    cycle list={
      {thick, blue,            mark=*,          mark size=1.6pt},
      {thick, red,  dashed,    mark=square*,    mark size=1.6pt},
      {thick, green!60!black, dotted, mark=triangle*, mark size=1.8pt},
    }
  },
  safeplot/.style={
    unbounded coords=jump,
    restrict x to domain=0:14
  }
}

\newcommand{\AddCatLabelsEbird}[1]{%
  \addplot+ [safeplot] table[x=z, y=#1] {\CatTableHiFiveEbird};
  \addplot+ [safeplot] table[x=z, y=#1] {\CatTableTippecanoeEbird};
  \addplot+ [safeplot] table[x=z, y=#1] {\CatTableAidEbird};
}

\newcommand{\AddGradLabelsEbird}[1]{%
  \addplot+ [safeplot] table[x=z, y=#1] {\GradTableHiFiveEbird};
  \addplot+ [safeplot] table[x=z, y=#1] {\GradTableTippecanoeEbird};
  \addplot+ [safeplot] table[x=z, y=#1] {\GradTableAidEbird};
}

\pgfplotsset{
  ebirdStyle/.style={
    unbounded coords=jump,
    restrict x to domain=0:14,
    grid=major,
    tick label style={font=\tiny},
    label style={font=\tiny},
    title style={font=\scriptsize, yshift=-1.5ex},
    xlabel={Zoom},
    enlarge x limits=false,
  }
}

\begin{tikzpicture}[scale=0.9]
\begin{groupplot}[
    group style={
        group size=2 by 3,         
        vertical sep=1.1cm,
        horizontal sep=1.1cm,
    },
    width=0.42\columnwidth,        
    height=0.35\columnwidth,
    scale only axis,
    ebirdStyle
]

\nextgroupplot[title={CAT: SSIM}, ylabel={SSIM}]
    \addplot[red, thick, mark=square*, mark size=1pt] table[x=z, y=mean_ssim] {\CatTableHiFiveEbird};
    \addplot[blue, mark=x, mark size=1.3pt] table[x=z, y=mean_ssim] {\CatTableAidEbird};
    \addplot[black, mark=o, mark size=1pt] table[x=z, y=mean_ssim] {\CatTableTippeEbird};

\nextgroupplot[title={GRAD: SSIM}]
    \addplot[red, thick, mark=square*, mark size=1pt] table[x=z, y=mean_ssim] {\GradTableHiFiveEbird};
    \addplot[blue, mark=x, mark size=1.3pt] table[x=z, y=mean_ssim] {\GradTableAidEbird};
    \addplot[black, mark=o, mark size=1pt] table[x=z, y=mean_ssim] {\GradTableTippeEbird};

\nextgroupplot[title={CAT: RMSE}, ylabel={RMSE}]
    \addplot[red, thick, mark=square*, mark size=1pt] table[x=z, y=mean_rmse] {\CatTableHiFiveEbird};
    \addplot[blue, mark=x, mark size=1.3pt] table[x=z, y=mean_rmse] {\CatTableAidEbird};
    \addplot[black, mark=o, mark size=1pt] table[x=z, y=mean_rmse] {\CatTableTippeEbird};

\nextgroupplot[title={GRAD: RMSE}]
    \addplot[red, thick, mark=square*, mark size=1pt] table[x=z, y=mean_rmse] {\GradTableHiFiveEbird};
    \addplot[blue, mark=x, mark size=1.3pt] table[x=z, y=mean_rmse] {\GradTableAidEbird};
    \addplot[black, mark=o, mark size=1pt] table[x=z, y=mean_rmse] {\GradTableTippeEbird};

\nextgroupplot[title={CAT: PSNR}, ylabel={PSNR}]
    \addplot[red, thick, mark=square*, mark size=1pt] table[x=z, y=mean_psnr] {\CatTableHiFiveEbird};
    \addplot[blue, mark=x, mark size=1.3pt] table[x=z, y=mean_psnr] {\CatTableAidEbird};
    \addplot[black, mark=o, mark size=1pt] table[x=z, y=mean_psnr] {\CatTableTippeEbird};

\nextgroupplot[title={GRAD: PSNR}]
    \addplot[red, thick, mark=square*, mark size=1pt] table[x=z, y=mean_psnr] {\GradTableHiFiveEbird};
    \addplot[blue, mark=x, mark size=1.3pt] table[x=z, y=mean_psnr] {\GradTableAidEbird};
    \addplot[black, mark=o, mark size=1pt] table[x=z, y=mean_psnr] {\GradTableTippeEbird};

\end{groupplot}

\matrix[
  matrix of nodes,
  anchor=north,
  nodes={font=\footnotesize, inner sep=1pt},
  column sep=6pt
] at ($(group c1r3.south)!0.5!(group c2r3.south) - (0,0.8cm)$) {
  \node{\tikz{\draw[red, thick] (0,0) -- (0.3,0); \fill[red] (0.12,-0.03) rectangle +(0.06,0.06);}}; & \node{Hi5}; &
  \node{\tikz{\draw[blue] (0,0) -- (0.3,0); \draw[blue] (0.12,-0.03) -- (0.18,0.03) (0.12,0.03) -- (0.18,-0.03);}}; & \node{AID*}; &
  \node{\tikz{\draw[black] (0,0) -- (0.3,0); \draw[black] (0.15,0) circle (0.04);}}; & \node{Tippe}; \\
};
\end{tikzpicture}
\caption{Per-zoom mean metrics for CAT vs GRAD render styles (eBird).}
\label{fig:ebird-metrics}
\end{figure}
}
\begin{figure}[t]
\centering


\pgfplotstableread{
 z  mean_ssim  mean_rmse  mean_psnr
 0   0.780902  45.940070  14.886971
 1   0.816083  38.961508  17.183405
 2   0.772037  50.347769  14.761022
 3   0.857523  46.204386  16.687021
 4   0.889084  46.003724  16.662720
 5   0.846260  44.610703  19.151280
 6   0.898874  37.417179  21.873629
 7   0.898385  34.949677  23.640131
 8   0.966008  19.856926  28.629486
 9   0.997118   5.993157  35.666923
10   0.999729   2.220182  41.054757
11   0.999941   0.761526  46.493801
12   0.999977   0.217519  51.601769
13   0.999991   0.047245  54.836429
}\CatTableHiFive

\pgfplotstableread{
 z  mean_ssim  mean_rmse  mean_psnr
 0   0.021436  58.826355  12.739365
 1   0.038746  48.908505  14.897462
 2   0.023376  66.196149  12.299277
 3   0.046393  64.971532  13.447445
 4   0.044996  66.613964  12.831697
 5   0.033556  65.029741  12.674822
 6   0.035644  62.287166  13.246814
 7   0.039405  59.193752  13.653544
 8   0.047937  53.029272  14.554246
 9   0.074649  48.532693  15.690881
10   0.093418  42.227277  17.114800
11   0.118614  36.415519  18.544418
12   0.156310  28.726783  20.470895
13   0.171473  23.759407  21.427306
}\CatTableTippecanoe

\pgfplotstableread{
 z  mean_ssim  mean_rmse  mean_psnr
 0   0.000003 244.008517   0.382704
 1   0.000004 246.311273   0.304138
 2   0.000003 240.147785   0.533234
 3   0.000004 239.243949   0.568266
 4   0.000004 237.099847   0.653247
 5   0.000004 235.375825   0.727310
 6   0.000005 234.286693   0.767996
 7   0.000004 235.591626   0.711084
 8   0.000005 239.957419   0.541221
 9   0.000005 244.246289   0.381031
10   0.000006 247.412105   0.265979
11   0.000007 249.516406   0.190141
12   0.000008 250.753876   0.146316
13   0.000007 251.163501   0.131845
}\CatTableAid

\pgfplotstableread{
 z  mean_ssim  mean_rmse  mean_psnr
 0   0.810615  39.729452  16.148552
 1   0.843024  29.502681  23.196799
 2   0.794637  41.297517  18.403115
 3   0.890373  30.260097  23.790285
 4   0.925093  24.176978  26.899464
 5   0.862178  31.684563  27.588529
 6   0.911075  23.782597  30.398285
 7   0.908366  22.083455  31.033846
 8   0.973765   8.664229  35.867159
 9   0.998164   2.244576  41.205077
10   0.999977   0.633948  46.868623
11   0.999994   0.193715  49.729197
12   0.999997   0.052398  53.523406
13   1.000000   0.002973  53.502570
}\GradTableHiFive

\pgfplotstableread{
 z  mean_ssim  mean_rmse  mean_psnr
 0   0.013212  71.372619  11.060171
 1   0.024553  59.357063  13.207958
 2   0.014972  80.154017  10.651854
 3   0.032063  78.685402  11.789709
 4   0.031085  80.761968  11.181764
 5   0.020872  78.861781  11.014970
 6   0.022553  75.670541  11.579186
 7   0.023919  72.036698  11.977106
 8   0.028940  64.438586  12.889001
 9   0.045783  58.887017  14.038394
10   0.060751  51.252730  15.441699
11   0.083757  44.039718  16.915715
12   0.113675  34.381359  18.964881
13   0.113146  28.316804  19.951912
}\GradTableTippecanoe

\pgfplotstableread{
 z  mean_ssim  mean_rmse  mean_psnr
 0   0.000002 238.310689   0.587933
 1   0.000002 241.813652   0.468604
 2   0.000002 232.080195   0.851860
 3   0.000003 230.542643   0.917745
 4   0.000003 226.925392   1.079903
 5   0.000003 223.490834   1.266681
 6   0.000003 221.761269   1.326880
 7   0.000003 224.351339   1.185848
 8   0.000003 231.708720   0.868708
 9   0.000003 238.536555   0.597685
10   0.000004 243.472079   0.410936
11   0.000004 246.709268   0.290293
12   0.000005 248.596694   0.222007
13   0.000004 249.194236   0.200469
}\GradTableAid

\pgfplotsset{
  myCycle/.style={
    cycle list={
      {thick, blue,            mark=*,          mark size=1.8pt},
      {thick, red,  dashed,    mark=square*,    mark size=1.8pt},
      {thick, green!60!black, dotted, mark=triangle*, mark size=2.0pt},
    }
  },
  safeplot/.style={
    unbounded coords=jump,
    restrict x to domain=0:14
  }
}
\begin{tikzpicture}[scale=0.9]
\begin{groupplot}[
    group style={
        group size=2 by 3,         
        vertical sep=0.6cm,
        horizontal sep=1.1cm,
    },
    width=0.42\columnwidth,
    height=0.35\columnwidth,
    scale only axis,
]

\nextgroupplot[
    title={CAT},
    ylabel={SSIM},
    xlabel={}
]
    \addplot[red, thick, mark=square*, mark size=1pt] table[x=z, y=mean_ssim] {\CatTableHiFive};
    \addplot[blue, mark=x, mark size=1.3pt] table[x=z, y=mean_ssim] {\CatTableAid};
    \addplot[black, mark=o, mark size=1pt] table[x=z, y=mean_ssim] {\CatTableTippecanoe};

\nextgroupplot[
    title={GRAD},
    ylabel={},
    xlabel={}
]
    \addplot[red, thick, mark=square*, mark size=1pt] table[x=z, y=mean_ssim] {\GradTableHiFive};
    \addplot[blue, mark=x, mark size=1.3pt] table[x=z, y=mean_ssim] {\GradTableAid};
    \addplot[black, mark=o, mark size=1pt] table[x=z, y=mean_ssim] {\GradTableTippecanoe};

\nextgroupplot[
    title={},
    ylabel={RMSE},
    xlabel={}
]
    \addplot[red, thick, mark=square*, mark size=1pt] table[x=z, y=mean_rmse] {\CatTableHiFive};
    \addplot[blue, mark=x, mark size=1.3pt] table[x=z, y=mean_rmse] {\CatTableAid};
    \addplot[black, mark=o, mark size=1pt] table[x=z, y=mean_rmse] {\CatTableTippecanoe};

\nextgroupplot[
    title={},
    ylabel={},
    xlabel={}
]
    \addplot[red, thick, mark=square*, mark size=1pt] table[x=z, y=mean_rmse] {\GradTableHiFive};
    \addplot[blue, mark=x, mark size=1.3pt] table[x=z, y=mean_rmse] {\GradTableAid};
    \addplot[black, mark=o, mark size=1pt] table[x=z, y=mean_rmse] {\GradTableTippecanoe};

\nextgroupplot[
    title={},
    ylabel={PSNR},
    xlabel={Zoom}
]
    \addplot[red, thick, mark=square*, mark size=1pt] table[x=z, y=mean_psnr] {\CatTableHiFive};
    \addplot[blue, mark=x, mark size=1.3pt] table[x=z, y=mean_psnr] {\CatTableAid};
    \addplot[black, mark=o, mark size=1pt] table[x=z, y=mean_psnr] {\CatTableTippecanoe};

\nextgroupplot[
    title={},
    ylabel={},
    xlabel={Zoom}
]
    \addplot[red, thick, mark=square*, mark size=1pt] table[x=z, y=mean_psnr] {\GradTableHiFive};
    \addplot[blue, mark=x, mark size=1.3pt] table[x=z, y=mean_psnr] {\GradTableAid};
    \addplot[black, mark=o, mark size=1pt] table[x=z, y=mean_psnr] {\GradTableTippecanoe};

\end{groupplot}

\matrix[
  matrix of nodes,
  anchor=north,
  nodes={font=\footnotesize, inner sep=1pt},
  column sep=6pt
] at ($(group c1r3.south)!0.5!(group c2r3.south) - (0,0.8cm)$) {
  \node{\tikz{\draw[red, thick] (0,0) -- (0.3,0); \fill[red] (0.12,-0.03) rectangle +(0.06,0.06);}}; & \node{Hi5}; &
  \node{\tikz{\draw[blue] (0,0) -- (0.3,0); \draw[blue] (0.12,-0.03) -- (0.18,0.03) (0.12,0.03) -- (0.18,-0.03);}}; & \node{AID*}; &
  \node{\tikz{\draw[black] (0,0) -- (0.3,0); \draw[black] (0.15,0) circle (0.04);}}; & \node{Tippecanoe}; \\
};
\end{tikzpicture}
\caption{Per-zoom mean metrics for CAT vs GRAD render styles (roads).}
\label{fig:roads-metrics}
\end{figure}
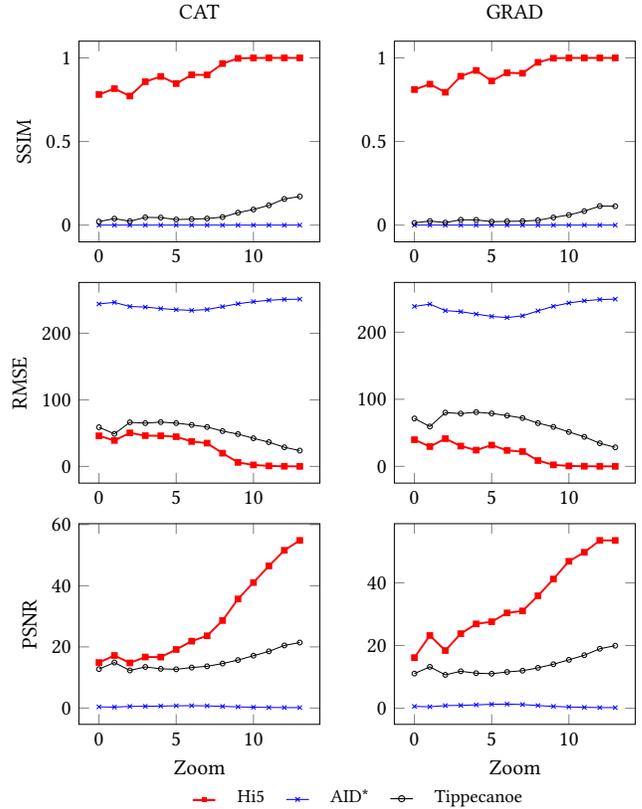

\begin{figure}
    \centering

\pgfplotstableread{
thrKB  styleA  styleB  styleC
32     51.90   55.37   48.57
64     40.63   51.10   48.08
128    18.16   45.44   47.82
256     1.56   36.82   47.35
512     0.15   30.61   43.27
1024    0.0     0.0    14.77
2048    0.0     0.0    6.66
4096    0.0     0.0    0.0
8192    0.0     0.0    0.0
16384   0.0     0.0    0.0
}\thrsweep

\begin{tikzpicture}
\begin{groupplot}[
    group style={
        group size=2 by 1,
        horizontal sep=0.8cm,
        xlabels at=edge bottom,
    },
    width=0.38\columnwidth,
    height=0.36\columnwidth,
    scale only axis,
    grid=major,
    grid style={dashed, gray!30},
    tick label style={font=\small},
    label style={font=\small},
    title style={font=\small, yshift=-1ex}
]

\nextgroupplot[
    xlabel={Cells ($\times 10^3$)},
    ylabel={Mean runtime (ms)},
    xmin=45, xmax=105,
    ymin=400, ymax=1350,
    xtick={50, 75, 100},
    legend columns=2,
    legend style={at={(0.0,1.0)}, anchor=north west,font=\scriptsize},
]
    \addplot[color=red, mark=square*, thick, mark size=1pt]
        coordinates {(50,490.20) (60,634.46) (70,684.17) (80,753.75) (90,1040.83) (100,1020.45)};
    \addplot[color=blue, mark=x, thick, mark size=1.5pt]
        coordinates {(50,455.25) (60,655.23) (70,674.00) (80,768.75) (90,878.83) (100,1015.60)};
    \addplot[color=black, mark=o, thick, mark size=1pt]
        coordinates {(50,479.40) (60,631.23) (70,684.80) (80,1021.25) (90,1041.00) (100,1007.44)};
    \addplot[color=gray, mark=triangle*, thick, mark size=1.2pt]
        coordinates {(50,482.20) (60,563.31) (70,795.40) (80,799.25) (90,993.67) (100,1005.65)};
    \legend{32KB,64KB,128KB,256KB}

\nextgroupplot[
    xmode=log,
    log basis x=2,
    xlabel={Tile size budget, $B$ (KB)},
    ylabel={Mean RMSE},
    xmin=32, xmax=16384,
    ymin=0, ymax=70,
    xtick={32,64,128,256,512,1024,2048,4096,8192,16384},
    legend style={at={(1.0,1.0)}, anchor=north east,font=\scriptsize},
]
    \addplot[blue, solid, thick, mark=o, mark size=2.3pt]
        table[x=thrKB, y=styleA] {\thrsweep};

    \addplot[red, solid, thick, mark=square, mark size=2.3pt]
        table[x=thrKB, y=styleB] {\thrsweep};

    \addplot[green!60!black, solid, thick, mark=triangle, mark size=2.3pt]
        table[x=thrKB, y=styleC] {\thrsweep};

    \legend{Counties,NE\_roads,eBird (1\%)}
\end{groupplot}

\end{tikzpicture}
    \caption{Left: Runtime vs. number of input cells for MILP solver, for different output size constraints. Right: Mean RMSE after reduction across top 20 tiles with highest initial byte size using different size constraints.
    }
    \label{fig:sparsification-cells-time}
\end{figure}
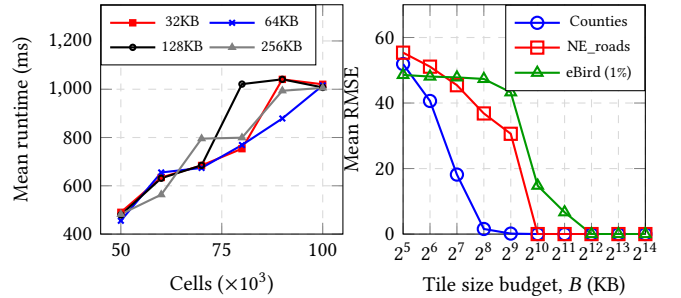

\begin{figure}
    \centering
    \pgfplotstableread{
alpha  styleA        styleB  styleC
0.0    5.47 65.02   49
0.2    1.89 31.30   43
0.4    1.83 28.10   43
0.6    1.64 30.26   43
0.8    1.55 30.80   43
1.0    5.09 35.80   48
}\alpharmse
\begin{tikzpicture}
\begin{groupplot}[
    group style={
        group size=3 by 1,
        horizontal sep=1.1cm
    },
    width=0.42\columnwidth,
    height=0.38\columnwidth,
    xlabel={$\alpha$},
    grid=major,
    tick label style={font=\scriptsize},
    label style={font=\small},
]

\nextgroupplot[
    title={Counties},
    ylabel={Mean RMSE},
    xmin=0, xmax=1,
    ymin=0, ymax=6
]
\addplot+[mark=o] table[x=alpha, y=styleA] {\alpharmse};

\nextgroupplot[
    title={Roads},
    xmin=0, xmax=1,
    ymin=25, ymax=70
]
\addplot+[mark=o] table[x=alpha, y=styleB] {\alpharmse};

\nextgroupplot[
    title={eBird},
    xmin=0, xmax=1,
    ymin=40, ymax=50
]
\addplot+[mark=o] table[x=alpha, y=styleC] {\alpharmse};

\end{groupplot}
\end{tikzpicture}
    \caption{Mean RMSE vs $\alpha$ for three datasets.}
    \label{fig:alpha_rmse}
\end{figure}
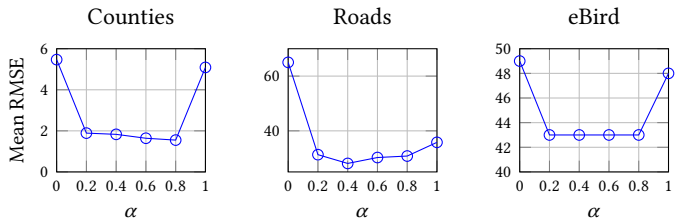

\mycomment{
\begin{figure}
    \centering

\pgfplotstableread{
thrKB  styleA  styleB  styleC
32     51.90   55.37   48.57
64     40.63   51.10   48.08
128    18.16   45.44   47.82
256     1.56   36.82   47.35
512     0.15   30.61   43.27
1024    0.0     0.0    14.77
2048    0.0     0.0    6.66
4096    0.0     0.0    0.0
8192    0.0     0.0    0.0
16384   0.0     0.0    0.0
}\thrsweep

\begin{tikzpicture}
\begin{axis}[
    width=0.95\linewidth,
    height=0.48\linewidth,
    xmode=log,
    log basis x=2,
    xlabel={Tile size constraint (KB)},
    ylabel={Mean RMSE},
    xtick={32,64,128,256,512,1024,2048,4096,8192,16384},
    xmin=32, xmax=16384, 
    ymin=0,
    grid=major,
    tick label style={font=\scriptsize},
    label style={font=\small},
    legend style={
        at={(0.5,-0.25)},     
        anchor=north,         
        font=\scriptsize,
        draw=none,
        fill=none,
        legend columns=-1,    
        column sep=10pt,      
    },
    legend cell align=left,
]

\addplot+[
    blue, solid, line width=1.5pt,
    mark=o, mark size=2.5pt
] table[x=thrKB, y=styleA] {\thrsweep};
\addlegendentry{Counties}

\addplot+[
    red, dashed, line width=1.5pt,
    mark=square, mark size=2.5pt
] table[x=thrKB, y=styleB] {\thrsweep};
\addlegendentry{NE\_roads}

\addplot+[
    green!60!black, dotted, line width=1.5pt,
    mark=triangle, mark size=2.5pt
] table[x=thrKB, y=styleC] {\thrsweep};
\addlegendentry{eBird 0.01}

\end{axis}
\end{tikzpicture}
\caption{Mean RMSE vs tile size constraint for three datasets/styles. Style A: County dataset, gradient styling based on water area. Style B: NE\_Roads\_NA dataset, gradient coloring based on length. 
}
\label{fig:size_rmse}
\end{figure}
}

\begin{table}[t]
\footnotesize
\caption{Spearman rank correlation between TLD and visualization error metrics across datasets. Higher TLD indicates larger deviation; thus correlations are positive for RMSE and negative for PSNR/SSIM.}
\centering
\renewcommand{\arraystretch}{1.15}
\begin{tabular}{lccc}
\hline
\textbf{Metric} & \textbf{Counties} & \textbf{NE\_Roads} & \textbf{ebird 0.01} \\
\hline
Spearman$(\mathrm{TLD}, \mathrm{RMSE})$ & 0.999823 & 0.813129 & 0.882601 \\
Spearman$(\mathrm{TLD}, \mathrm{PSNR})$ & -0.951515 & -0.813129 & -0.882601 \\
Spearman$(\mathrm{TLD}, \mathrm{SSIM})$ & -0.999867 & -0.797708 & -0.856090 \\
\hline
\end{tabular}
\label{tab:tld_spearman_corr_datasets}
\end{table}

To evaluate the utility of the tile-level distortion function (TLD) in \autoref{eq:tld}, \autoref{tab:tld_spearman_corr_datasets} shows the Spearman correlation coefficient between the three quality metrics above and the proposed TLD metric. We notice that TLD exhibits a very strong correlation with all three metrics when applied on stylized images. This is a very important finding because TLD is a style-free metric but it can be used in lieu of evaluating styled images. This means that researchers working on this problem in the future can focus on solving the proposed visualization aware tile reduction problem without having to tailor their method for a specific style.

\subsection{Sparsification Experiments}
All MILP instances are solved using the Gurobi Optimizer (11.0.2) \cite{gurobi} with a $1\%$ relative optimality gap.
This section studies the performance of the sparsification step and how it affects the quality of the generated tiles. To measure its performance, we run the full pyramid generation for eight levels and record each invocation of sparsification step (\autoref{alg:cell-sparsify}). \autoref{fig:sparsification-cells-time} on the left reports the average time as the number of cells ($N\times d$) varies from 50K to 100K. This range includes tiles that are larger than the budget ($B$) to trigger the sparsification algorithm and lower than the reservoir capacity of the triage step ($100$k cells).
As shown in figure, the running time grows roughly linear with the number of cells since it dictates the number of variables in the problem. In addition, it is observed that the tile budget $B$ has little to no effect on the overall running time since it does not affect the size of the problem. 



The right side of \autoref{fig:sparsification-cells-time} reports the mean error the largest 20 tiles for three datasets containing points, lines, and polygons. We vary the tile size budget ($B$) from 32~kb to 4~mb and measure its effect on the visualization error (RMSE). As expected, the error drops rapidly as the budget $B$ increases. This experiment gives guidance for system designers on how to choose the tile budget. Another dimension that affects the choice is the client-side capability and loading time but this is out of scope of our study.

Finally, \autoref{fig:alpha_rmse} illustrates the effect of the parameter $\alpha$ in \autoref{eq:milp-objective} on output tile quality. Recall that $\alpha$ balances the visual prominence, measured by the pixel footprint, and semantic information, quantified by the KL-divergence. The results show that the two extreme cases ($\alpha=0$ and $\alpha=1$) yield the worst performance as it focuses on one aspect only. It also shows that there is a wide range of values in which $\alpha$ can be set which makes it relatively simple for system designers to set this value.


\mycomment{
\begin{figure}
\pgfplotstableread{
Metric          Polygons   Lines   Points
Rand            16    50.4    48
Vert            12.9   43.6   49
Pix             12.7    43.7    49
BS(e)           15.6    44.1    46
BS(a)           13.1    44.4    49
Attr            15.1    45.7    47
}\rmseData

\pgfplotstableread{
Metric          Polygons  Lines   Points
Rand            149     66     10
Vert            146     72     10
Pix             159     75     20
BS(e)           186     66     16
BS(a)           142     71     12
Attr            146     66     11
}\timeData

\begin{tikzpicture}
    \pgfplotsset{
        my plot style/.style={
            width=0.95\columnwidth, 
            height=3.8cm,
            ybar=0pt,
            bar width=0.07in,
            enlarge x limits=0.15,
            xtick=data,
            tick label style={font=\tiny},
            label style={font=\tiny},
            grid=major,
            grid style={dashed, gray!30},
            ymin=0,
        }
    }

    \begin{axis}[
        name=toplot,
        my plot style,
        title={\scriptsize RMSE Comparison},
        ylabel={RMSE},
        xticklabels={}, 
    ]
        \addplot[fill=yellow!85!orange, draw=orange!60!black, postaction={pattern=north east lines}] 
            table[x expr=\coordindex, y=Points]{\rmseData};
        \addplot[fill=cyan!80, draw=cyan!50!black, postaction={pattern=crosshatch}] 
            table[x expr=\coordindex, y=Lines]{\rmseData};
        \addplot[fill=magenta!70, draw=magenta!50!black, postaction={pattern=dots}] 
            table[x expr=\coordindex, y=Polygons]{\rmseData};
    \end{axis}

    \begin{axis}[
        name=botplot,
        at={($(toplot.south) - (0,0.6cm)$)}, 
        anchor=north,
        my plot style,
        title={\scriptsize Runtime Comparison},
        title style={yshift=-1.2ex}, 
        ylabel={Time (ms)},
        xticklabels={Rand, Vert, Pix, BS(e), BS(a), Attr},
        xticklabel style={rotate=45, anchor=north east},
    ]
        \addplot[fill=yellow!85!orange, draw=orange!60!black, postaction={pattern=north east lines}] 
            table[x expr=\coordindex, y=Points]{\timeData};
        \addplot[fill=cyan!80, draw=cyan!50!black, postaction={pattern=crosshatch}] 
            table[x expr=\coordindex, y=Lines]{\timeData};
        \addplot[fill=magenta!70, draw=magenta!50!black, postaction={pattern=dots}] 
            table[x expr=\coordindex, y=Polygons]{\timeData};
    \end{axis}

    \matrix[
      matrix of nodes,
      anchor=north,
      nodes={font=\footnotesize, inner sep=1pt},
      column sep=5pt
    ] at ($(botplot.south) - (0,0.8cm)$) {
      \node{\tikz{\draw[fill=yellow!85!orange, postaction={pattern=north east lines}] (0,0) rectangle (0.2,0.1);}}; & \node{Points (eBird 0.001)}; &
      \node{\tikz{\draw[fill=cyan!80, postaction={pattern=crosshatch}] (0,0) rectangle (0.2,0.1);}}; & \node{Lines (NE Roads NA)}; &
      \node{\tikz{\draw[fill=magenta!70, postaction={pattern=dots}] (0,0) rectangle (0.2,0.1);}}; & \node{Counties (Polygons)}; \\
    };
\end{tikzpicture}
\caption{Impact of different prioritization metrics on visual fidelity and computation time.}
\label{fig:triage-priorities}
\end{figure}
}

\mycomment{\subsection{Triage Experiments}}

As seen in fig \ref{fig:sparsification-cells-time} the number of cells input to the sparsification step is the driving factor for the runtime of the MILP solver, so we set a fixed capacity of cells for tiles at the Triage step and aim to fill this budget of 100k cells with different possible priorities, as discussed in section \ref{sec:triage}. 
\mycomment{In fig \ref{fig:triage-priorities} the effect of using different priorities on three datasets is shown. The runtime is similar for number of vertices and number of (non-null) attributes as calculating these metrics for each record is more straightforward, and computing number of pixels requires more time, but the highest time belongs to byte size (exact) calculation as records need to be converted to byte arrays. An approximate method is also used for the record byte size which takes less time for datasets that contain more complex geometries. According to the geometry type of the dataset, a different priority for record-level triage may be most suitable, for example, pixel count (or vertices count as a faster proxy) for datasets that consist of polygons of varying sizes result in better quality output tiles. In points datasets, as all records have the same count of vertices (one) and pixel coverage, these priorities behave similar to random, and using record byte size as priority will result in keeping more records in the tile.}

\subsection{Qualitative Evaluation}

\autoref{fig:tile-qualitative-3tiles-1col} shows two representative tiles rendered with four methods. The first row depicts the Counties dataset with gradient coloring by the \texttt{AWATER} attribute. Under the same tile size constraint, both Tippecanoe and \hifive closely match the original, though the red circles highlight regions where Tippecanoe drops or coalesces records.
The second row shows the Roads dataset with categorical coloring based on the ``continent" attribute. Here, \hifive removes this attribute's values for some smaller records, causing them to appear in black, while Tippecanoe removes more records altogether (i.e., any retained record preserves all of its attributes). As a result, Tippecanoe's output is sparser.


\begin{figure}[t]
\centering
\setlength{\tabcolsep}{0pt}

\newcommand{\imgw}{0.24\columnwidth} 
\newcommand{\gap}{1pt}
\newcommand{\rowgap}{1.0mm}
\begin{tabular}{@{}c@{\hspace{\gap}}c@{\hspace{\gap}}c@{\hspace{\gap}}c@{}}

Ground Truth & Hi5 & Tippecanoe & AID* \\

\mycomment{
\includegraphics[width=\imgw]{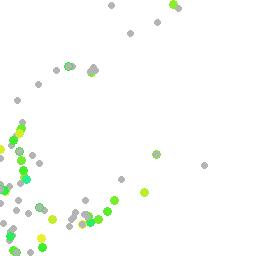} &
\includegraphics[width=\imgw]{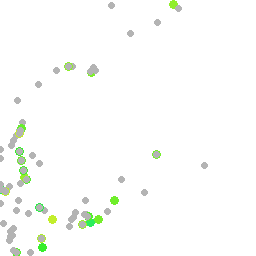} &
\includegraphics[width=\imgw]{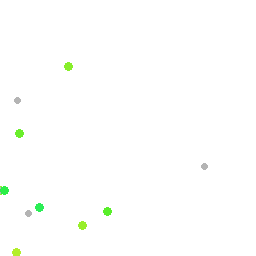} &
\includegraphics[width=\imgw]{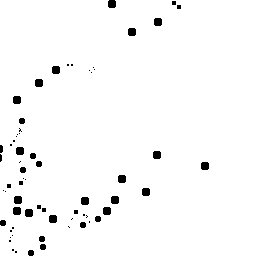} \\
\vspace{\rowgap} \\[-\rowgap]
}

\includegraphics[width=\imgw]{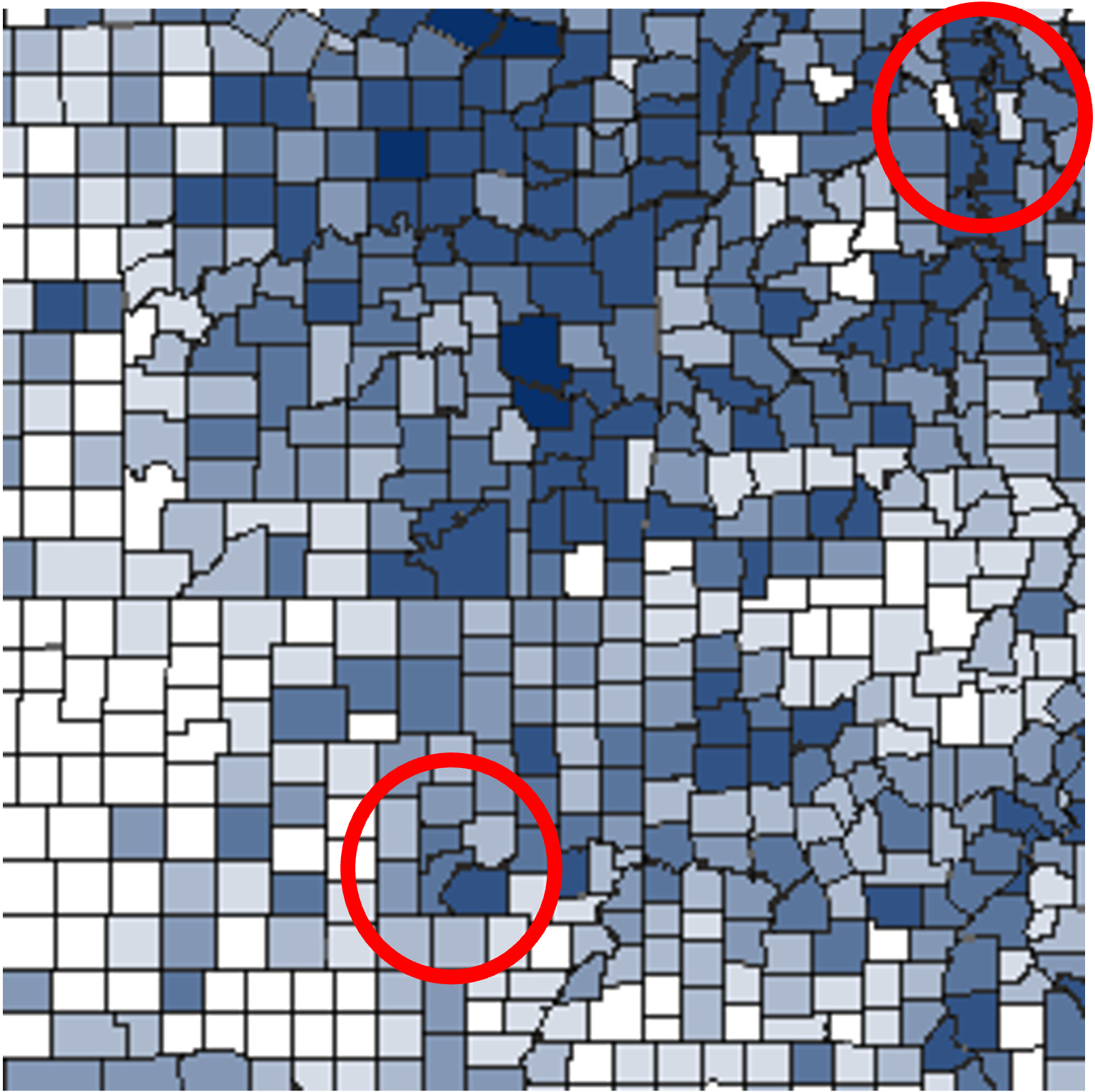} &
\includegraphics[width=\imgw]{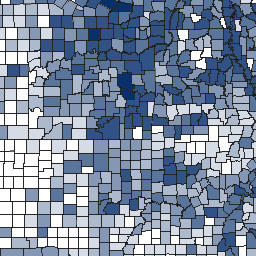} &
\includegraphics[width=\imgw]{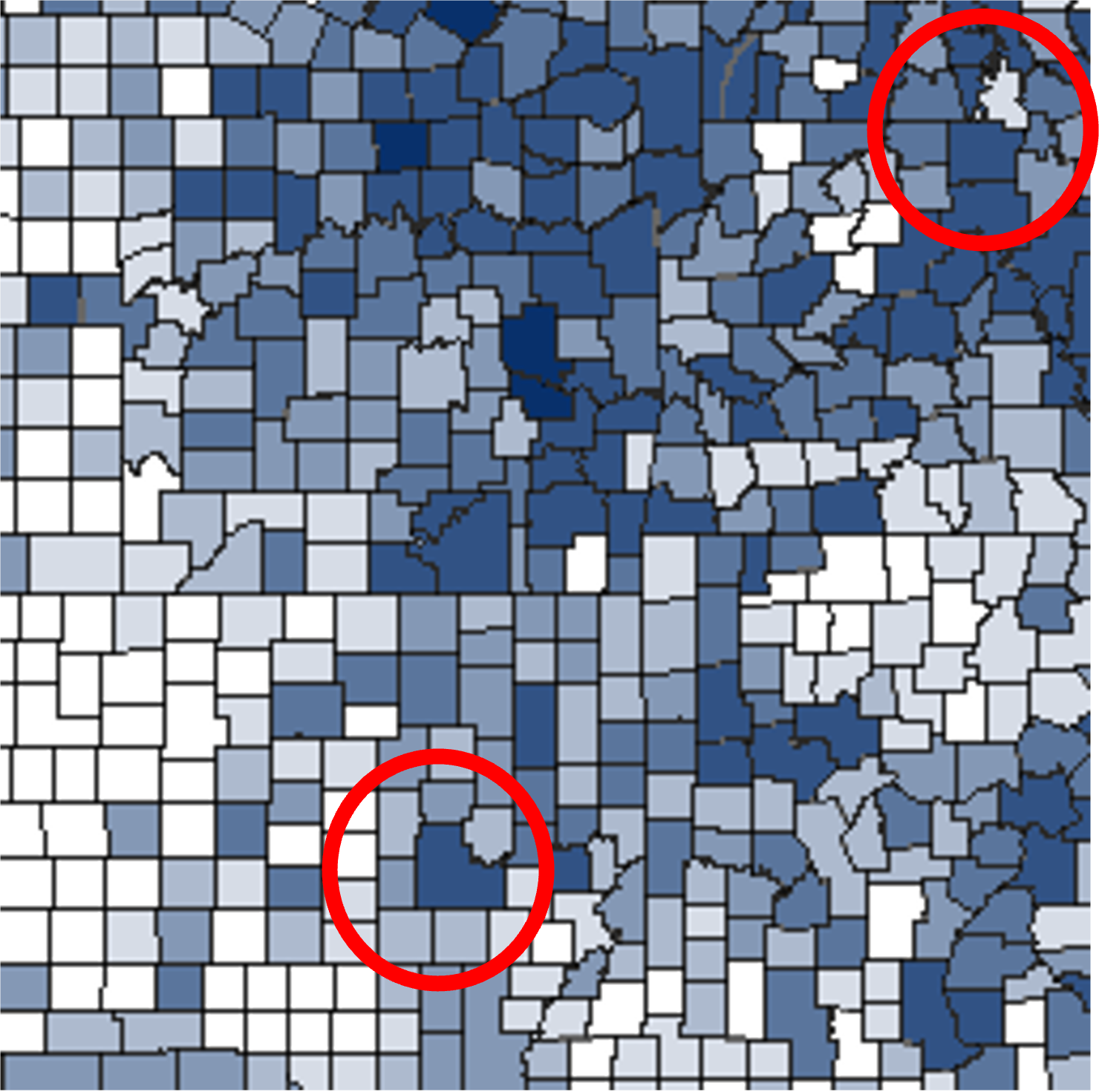} &
\includegraphics[width=\imgw]{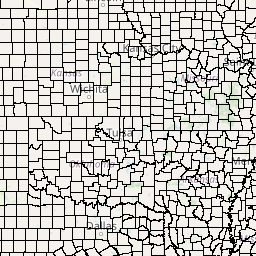} \\
\vspace{\rowgap} \\[-\rowgap]

\includegraphics[width=\imgw]{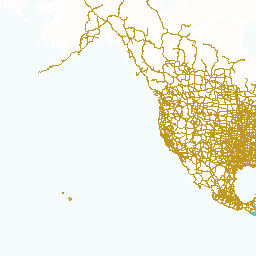} &
\includegraphics[width=\imgw]{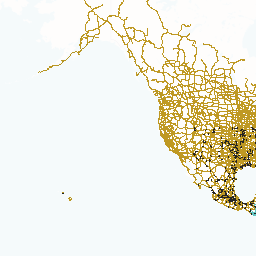} &
\includegraphics[width=\imgw]{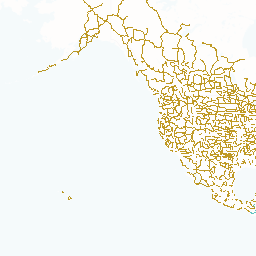} &
\includegraphics[width=\imgw]{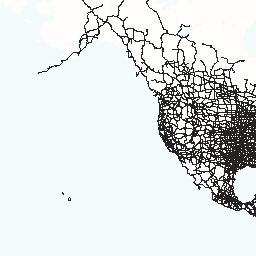} \\

\end{tabular}

\caption{Qualitative comparison for two tiles. Each row shows the same stylized tile with ground truth followed by Hi5, Tippecanoe, and AID*}
\vspace{-15pt}
\label{fig:tile-qualitative-3tiles-1col}
\end{figure}

\balance

\section{Conclusion}
\label{sec:conclusion}
This paper presented \textsc{HiFIVE}, a visualization-aware framework for reducing vector-tile sizes while preserving both visual salience and semantic fidelity for interactive map exploration.
By formulating tile reduction as an optimization problem, \textsc{HiFIVE} explicitly captures the trade-off between storage cost and visualization quality, jointly reasoning about geometric prominence and attribute information loss.
Our MILP-based sparsification selectively removes records and attribute values using a linear size model grounded in geometry complexity and dictionary-based encoding, while cell utilities derived from information-theoretic divergence ensure that retained values minimally alter attribute distributions.
Combined with a lightweight triage stage for scalability, \textsc{HiFIVE} produces compact, standard-compliant vector tiles that support client-side styling.
Given that the proposed style-free Tile-Level Distortion (TLD) metric correlates strongly with perceptual image quality, future work should explore alternative optimization algorithms that minimize this distortion.
\nocite{*}
\bibliographystyle{ACM-Reference-Format}
\bibliography{references}

\end{document}